\documentclass[notitlepage,nofootinbib]{revtex4-1}
\pdfoutput=1
\usepackage[utf8]{inputenc}
\usepackage[english]{babel}
\usepackage[T1]{fontenc}
\usepackage{amssymb,amsthm,amsfonts,graphicx,hyperref,parskip,setspace,float,caption}
\usepackage[inline]{enumitem}
\usepackage{thmtools,thm-restate}
\usepackage[dvipsnames,svgnames,x11names]{xcolor}
\usepackage{amsmath} 
\usepackage[a4paper,left=2.5cm,right=2.5cm,head=4cm,bottom=2.5cm]{geometry}

\newtheorem{lemma}{Lemma}
\theoremstyle{remark}
\newtheorem*{remark}{Remark}
\numberwithin{equation}{section}

\def\rket{\rangle}

\def\lbra{\langle}
\newcommand{\bra}[1]{\mathinner{\left \langle #1 \right\rvert}}
\newcommand{\ket}[1]{\mathinner{\left \lvert #1 \right\rangle}}

\newcommand\lr[1]{\left( #1 \right)}
\newcommand\lrv[1]{\left|  #1 \right|}

\newcommand\lrb[1]{\left\lbrace #1 \right\rbrace}

\newcommand{\mbb}[1]{\mathbb{#1}}

\newcommand{\mc}[1]{\mathcal{#1}}
\newcommand{\vbl}{\vphantom{X^{X^X}}}
\newcommand{\D}{\,\mathrm{d}}
\newcommand{\dx}{\D x}
\newcommand{\dy}{\D y}
\newcommand{\ds}{\D s}
\newcommand{\dt}{\D t}
\newcommand{\dz}{\D z}

\newcommand{\hankel}[1] [\nu]{H^{(1)}_{#1}}

\newcommand{\e}{\mathrm{e}}
\newcommand{\imagI}{ \mathrm{i} }
\newcommand{\ratioLower}{\alpha}

\newcommand{\constFinal}{1.4818}
\newcommand{\constT}{0.8663}
\newcommand{\constC}{0.13368}
\newcommand{\constA}{0.7326}

\usepackage[capitalise]{cleveref}

\makeatletter
\def\blfootnote{\gdef\@thefnmark{}\@footnotetext}
\makeatother

\begin{document}
\title{Strong dispersion property for the quantum walk on the hypercube}

\author{M Kokainis$^1$, K Pr\=usis$^1$,  J Vihrovs$^1$, V Kashcheyevs$ ^2 $ and A Ambainis$^1$}

\address{$^1$ Centre for Quantum Computer Science, Faculty of Computing, University of Latvia, Rai\c{n}a 19, Riga, Latvia, LV-1586}
\address{$^2$ Department of Physics, University of Latvia, Jelgavas 3, Riga, Latvia, LV-1004}

\begin{abstract}
  We show that the discrete time quantum walk on the Boolean hypercube of dimension $n$ has a strong dispersion property: if the walk is started in one vertex, then the probability of the walker being at any particular vertex after $O(n)$ steps is of an order $O(1.4818^{-n})$. This improves over the known mixing results for this quantum walk which show that the probability
  distribution after $O(n)$ steps is close to uniform but do not
  show that the probability is small for every vertex.
  Our result shows that quantum walk on hypercube is interesting for algorithmic applications which require 
fast dispersion over the state space.
\end{abstract}

\maketitle


\blfootnote{}
\blfootnote{This is an Accepted Manuscript. The journal reference for the final published version is \href{http://doi.org/10.1088/1751-8121/aca6b9}{\textit{J. Phys. A: Math. Theor.} \textbf{55}, 495301 (2022)}. This Accepted Manuscript is available for reuse under a \href{https://creativecommons.org/licenses/by-nc-nd/4.0/}{CC BY-NC-ND licence} after the 12 month embargo period provided that all the terms and conditions of the licence are adhered to.}

\section{Introduction}
Quantum walks are the quantum counterpart of random walks. They have been very useful for designing quantum algorithms, from the exponential speedup for the ``glued trees'' problem by Childs et al.~\cite{childs2003exponential} and the element distinctness algorithm of \cite{ambainis2007element} to general results about speeding up classes of Markov chains \cite{szegedy2004Markov,apers2019fast,ambainis2020quadratic,apers2021unified}. Quantum walks also have applications to other areas (e.g. quantum state transfer \cite{mohseni2008environment} or Hamiltonian complexity \cite{nagaj2009fast}) and are interesting objects of study on their own. 

One of most important properties of both classical random walks and quantum walks is {\em rapid mixing} \cite{levin2017markov,randall2006rapidly}: if the walk is started in one vertex, after a certain number of steps the probability distribution of the walker is almost uniformly distributed over all vertices. Rapid mixing takes place for many graphs, the key condition for it is that the graph on which the walker is walking has no bottlenecks that may slow it down.  Classically, rapid mixing is useful for a variety of algorithms that perform sampling, counting or integration (for example, algorithms for estimating volumes of convex bodies \cite{dyer1991random,lovasz2006simulated}). Quantum walks also mix rapidly in many cases  (e.g. \cite{moore2002quantum}), for an appropriate definition of mixing and their mixing times can be related to the same combinatorial quantities of the underlying graph as classically \cite{aharonov2001quantum}. 

In this paper, we show that a popular quantum walk, the discrete time walk on the hypercube \cite{moore2002quantum}, has a property that
is substantially stronger than standard mixing.
In more detail, the Boolean hypercube consists of $2^n$ vertices indexed by $n$ bit strings $x_1\ldots x_n, x_i\in\{0, 1\}$,with vertices $x_1\ldots x_n$ and $y_1\ldots y_n$ connected by an edge if the strings $x_1\ldots x_n$ and $y_1\ldots y_n$ differ in exactly one symbol. Quantum walk on the hypercube has been studied in detail \cite{alagic2005decoherence, krovi2006hitting,marquezino2008mixing,potovcek2009optimized} and it was the first graph for which a search algorithm by quantum walk was developed, by Shenvi et al. \cite{shenvi2003quantum}. 

We show that the discrete time quantum walk on the hypercube has a very strong dispersion property: if a walker is started in one vertex (with the coin register being in a uniform superposition of all possible directions), then, after $O(n)$ steps, the probability of the walker being in each vertex becomes exponentially small. Our computer simulations show that, after $0.849... n$ steps, the probability of being at each vertex is at most $1.93...^{-n}$. Since the hypercube has $2^n$ vertices, the quantum walk is close to achieving the biggest possible dispersion. (In the uniform distribution, each vertex has the probability $2^{-n}$.) Rigorously, we show a result of a similar asymptotic form with somewhat weaker constants: the probability of being at each vertex after $\constT n$ steps is at most $O(\constFinal^{-n})$.

These two results provide a stronger bound on the probabilities of individual vertices than the previously known mixing results which only imply that probabilities of most vertices are close to $1/2^n$ but do not exclude the possibility that some vertices may have significantly larger probability. For example, the mixing result of \cite{moore2002quantum} implies that the maximum probability of one vertex of the hypercube is $o(n^{-7/6})$. Our result improves on this exponentially. 

Up to our knowledge, a strong dispersion property like ours has not been known for any discrete time quantum walk. In continuous time, a perfect dispersion can be achieved. Namely \cite{moore2002quantum}, after $\frac{\pi}{4} n$ time, the probability distribution of the continuous-time walker over the vertices of the hypercube is exactly uniform.

There are two important distinctions between discrete and continuous time walks here, one from applications perspective, one from methods perspective. From an applications perspective, quantum algorithms consist of discrete time steps. Hence, discrete time quantum walks are more suited to being used in a quantum algorithm. 
In particular, we plan to explore quantum walks in the context of query problems that show exponential separation between quantum and randomized classical computation, along the lines of \cite{aaronson2018forrelation,brandao2013exponential}.
Having the dispersion property for discrete time quantum walks is essential in this context.

From a methods perspective, dispersion for continuous time walks is easy to prove, because the Hamiltonian of the continuous time walk can be expressed as a sum of Hamiltonians corresponding to each dimension of the hypercube. Thus, the dispersion property for the walk in $n$ dimensions follows from a similar property in one dimension which reduces to analyzing 2-by-2 matrices.

In contrast, the matrix of the discrete time walk does not factorize into parts corresponding to each dimension and this makes the result for the discrete time walk much more challenging. As a result, we need a sophisticated proof based on analytic properties of  Bessel functions. 

\section{Results}

For a positive integer $n$, 
let $[n]$ denote the set $\{1, \ldots, n\}$.

The $n$-dimensional Boolean hypercube $Q_n$ is a graph with $2^n$ vertices indexed by $x\in\{0, 1\}^n$ and edges $(x, y)$ for $x, y$ that differ in one coordinate. The notation $\{0\}^n$ is shortened as $0^n$. Let $|x|$ denote the Hamming weight of a vertex $x$, defined as the number of $i\in[n]$ with $x_i=1$. 
For $x\in\{0, 1\}^n$ and $i\in [n]$, $x^{(i)}$ denotes the vertex obtained from $x$ by changing the $i^{\rm th}$ component:
\begin{equation}
     x^{(i)} = (x_1, \ldots, x_{i-1}, 1-x_i, x_{i+1}, \ldots, x_n) .
\end{equation}

We consider the standard discrete-time quantum walk on the hypercube \cite{moore2002quantum}, with the Grover diffusion operator as the coin flip. The state space of this walk has basis states 
$\ket{x, i}$ where $x\in\{0, 1\}^n$ is a vertex of $Q_n$ and {$i\in[n]$} is an index for one of the directions for the edges of the hypercube.

One step of the quantum walk consists of two parts:
\begin{enumerate}
\item
We apply the diffusion transformation $D_n$ defined by 
$D_n\ket{x,i} = -\frac{n-2}{n} \ket{x,i} + \frac{2}{n} \sum_{j\in [n] \setminus \{ i\}} \ket{x,j}$ for all $x\in \{0,1\}^n$ and $i\in[n]$.
We refer to this step as {\em coin flip}, as it corresponds to a coin flip in a classical random walk which chooses direction $i$ in which a classical random walker proceeds.
\item
We apply the shift transformation $S$ defined by $S\ket{x, i}=
\ket{x^{(i)}, i}$.
\end{enumerate}

We denote the sequence of these two transformations by $W=S\,D_n$.
The walk is started in the state $\ket{\psi_{start}} = \sum_{i=1}^n \frac{1}{\sqrt{n}} \ket{0^{n}, i}$ where the walker is localized in one vertex and the direction register is in the uniform superposition of all possible directions.

Let
\begin{equation*}
    \ket{\psi^{(t)}} = \sum_{x\in\{0, 1\}^n, i\in[n]} 
\alpha^{(t)}_{x, i} \ket{x, i}
\end{equation*}
be the state of the quantum walk after $t$ steps
and $P(x, t)=\sum_i |\alpha^{(t)}_{x, i}|^2$
be the probability of the walker being at location $x$ at this time.

After $t\approx 0.85 n$ steps, the walker disperses over the vertices $x\in\{0, 1\}^n$ very well, with no particular vertex $x$ having a substantial probability $P(x, t)$ of the walker being there. 

In \cref{fig:figure1_AD}, the upper   panel 
shows the maximum probability of a single vertex $\max_x P(x, t)$ at every step of a  quantum walk on a 50-dimensional hypercube. We can see that this probability reaches a minimum of about $10^{-14}$ (which is only slightly larger than the  theoretical minimum of $2^{-50} = 0.88... \cdot 10^{-16}$) after a number of steps that is slightly less than $n=50$.

\begin{figure}[H]
		\centering
		\includegraphics[width=\textwidth,height=\textheight,keepaspectratio]{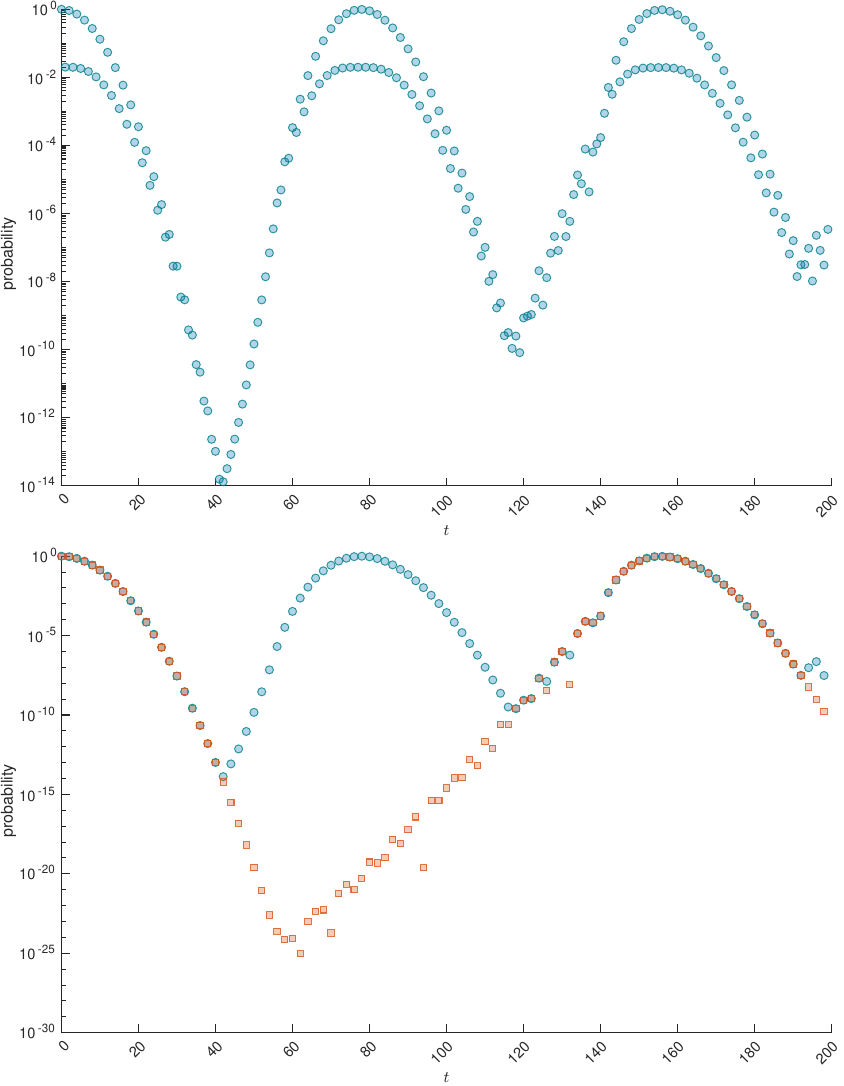}
		\caption{
		Illustration of the probabilities
		$\max_x P(x, t)$  and $ P(0^{50},t) $
		for a quantum walk on the 50-dimensional hypercube.\\
		\emph{Upper panel.} The maximum probability amongst all vertices $\max_x P(x, t)$.\\ 
		\emph{Lower panel.}  Only even steps   are shown. The circular markers:    the maximum probability amongst all vertices $\max_x P(x, t)$  (same as in the upper panel for even $t$). The square markers:  the probability at the initial vertex,  $ P(0^{50},t) $. \\
The data that support the graphs of this figure are available in the Zenodo repository \url{https://doi.org/10.5281/zenodo.5907185}.
		}
		\label{fig:figure1_AD}
\end{figure}

\begin{figure}[H]
		\centering
		\includegraphics[width=\textwidth,height=\textheight,keepaspectratio]{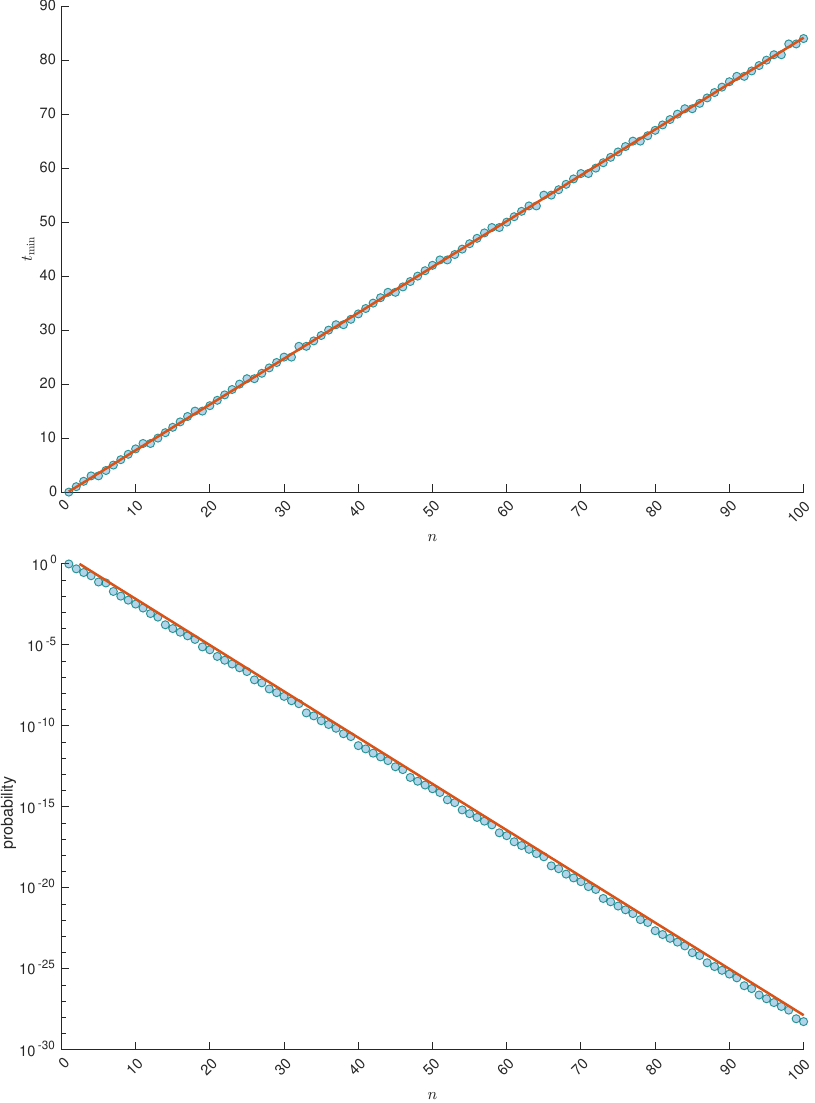}
		\caption{
		Illustration of $t_{\min}$ and  $\max_x P(x, t_{\min})$  for various values of $n$.\\ 
		\emph{Upper panel.} The circular markers: the number of steps $t_{\min}$ to reach the minimum of $\max_x P(x, t)$. The solid line: the graph of the function $-0.754+0.849n$, which approximates $t_{\min}$.\\
		\emph{Lower panel.} The circular markers:  the maximal (over all vertices of the hypercube)  probability $\max_x P(x, t_{\min})$, at time $t_{\min} \approx 0.849n$. The solid line:  the graph of the function $5\cdot 1.93^{-n}$, which is an upper bound on the probability.\\
The data that support the graphs of this figure are available in the Zenodo repository \url{https://doi.org/10.5281/zenodo.5907185}.
		}
		\label{fig:figure2_BC}
\end{figure}

We note that $\max_x P(x, t)$ fluctuates between odd and even numbered steps, due to the walker being at an odd distance from the starting vertex $0^n$ after an odd number of steps and at an even distance from $0^n$ after an even number of steps. This effect is particularly pronounced when $\max_x P(x, t)$ is large. Then, a large fraction of the probability is concentrated on $0^n$ after even steps but, after odd steps, this probability is equally divided among $n$ vertices $x$ with $|x|=1$. As a result, the maximum $\max_x P(x, t)$ is visibly larger after an even number of steps. This effect becomes smaller when $\max_x P(x, t)$ is small.

The upper panel of \cref{fig:figure2_BC} shows that the number of steps $t_{\min}$ required to minimize $\max_x P(x, t)$ grows linearly with $n$, and is achieved at approximately $-0.754+0.849 n$.
The probability $\max_x P(x, t_{\min})$ achieved at this $t_{\min}$ is approximately $5 \cdot 1.93^{-n}$ (see the lower   panel of \cref{fig:figure2_BC}). 

Lastly, for $t\leq t_{\min}$
this maximum is achieved at $x=0^n$ (at even numbered steps) or at vertices $x$ with $|x|=1$ (for odd numbered steps). 
This is shown in the lower   panel of \cref{fig:figure1_AD} where we plot $\max_x P(x, t)$ and $P(0^n, t)$ for even numbered steps $t$ (as the probability at $0^n$ is $0$ at odd steps).
The panel shows that the maximum is achieved at $x=0^n$
until the moment when the probability $\max_x P(x, t)$ starts increasing again. 

Rigorously, we can prove a weaker bound:
\begin{restatable}{theorem}{themainthm}\label{thm:main}
For any integer $t \in \lr{0.86628n,0.86632n} $, we have $\max_x 
P(x,t) = O(\constFinal^{-n})$.
\end{restatable}

\begin{remark}
 Even though we consider the case when the initial   state  is $ \ket{\psi_{start}}$,
 the symmetry of the walk implies a similar result when starting in a state 
 $ \sum_{i=1}^n \frac{1}{\sqrt{n}} \ket{x, i}$, for any hypercube vertex $x$.  This gives the following conclusion  from \cref{thm:main}: for all $x,y \in \{0,1\}^n $ and $t$ as in \cref{thm:main},
 \begin{equation}
 \frac{1}{n}
    \sum_{j=1}^n  \lrv{  
    \sum_{i=1}^n  \bra{y,j} W^t \ket{x, i} }^2
     =
      O(\constFinal^{-n}).
 \end{equation}
\end{remark}

\section{Proof of the main result}\label{sec:proof}
We first describe the strategy of the proof of \cref{thm:main}. Because of the symmetry of the walk, all vertices $x$ with the same $|x|=k$ will have equal probabilities $P(x, t)$, see \eqref{eq:symm} below. Since there are ${\binom{n}{k}}$ vertices $x$, this immediately implies $P(x, t)\leq 1/{\binom{n}{k}}$. If $k$ is such that ${\binom{n}{k}}$ is sufficiently large, we get the desired upper bound on $P(x, t)$. 

It remains to handle the case when ${\binom{n}{k}}$ is small. This corresponds to $k$ being either close to 0 ($0\leq k < \constC n$) or close to $n$ ($(1-\constC) n < k \leq n$). The second case is trivial: the number of time steps $t$ that we are considering is less than $(1-\constC) n$, so, vertices $x$ with $|x|>(1-\constC) n$ cannot be reached in $t$ steps. 

For the first case, we show (\cref{eq:lemma}) that if $P(x, t)$ is large, then $P(0^n, t')$ must also be non-negligible for some $t'\in\{t-|x|, \ldots, t+|x|\}$. Therefore, one can show an upper bound on all $P(x, t)$ with $|x| \leq \constC n$ by 
upper bounding $P(0^n, t')$ for all $t': t- \constC n \leq t' \leq t + \constC n$. This is done by \cref{thm:P0t} and 
\cref{thm:L11}, first expressing $P(0^n, t')$ in terms of Chebyshev polynomials and then bounding their asymptotics.

We begin by describing the evolution of the quantum walker in terms of states that utilize the symmetry of the   walk.
Denote
\begin{align}
\ket{w,\rightarrow} &= \frac{1}{\sqrt{{\binom{n}{w}}(n-w)}}\sum_{\substack{x \\ |x| = w}} \sum_{\substack{i \\ x_i = 0}} \ket{x,i}, \\
\ket{w,\leftarrow} &= \frac{1}{\sqrt{{\binom{n}{w}} w}}\sum_{\substack{x \\ |x| = w}} \sum_{\substack{i \\ x_i = 1}} \ket{x,i}.
\end{align}
By symmetry, the state of the quantum walk after any number of steps $t$ is of the form
\begin{equation}
    \ket{\psi^{(t)}} = 
\sum_{w=0}^{n-1}\alpha^{(t)}_{w, \rightarrow} \ket{w,\rightarrow} 
+
\sum_{w=1}^{n}\alpha^{(t)}_{w, \leftarrow} \ket{w,\leftarrow} .
\end{equation}
Let 
\begin{equation}
\label{eq:pwt}
    P[w,t] = |\alpha^{(t)}_{w, \rightarrow}|^2+ |\alpha^{(t)}_{w, \leftarrow}|^2
\end{equation}
be the total probability of the walker being
at one of vertices $x$ with $|x|=w$ after $t$ steps.
By the symmetry of the quantum walk, 
\begin{equation}
\label{eq:symm}
P(x,t)=P[w,t] / {\binom{n}{w}}
\end{equation} for any $x : |x|=w$. In particular, $P(0^n,t)=P[0,t]$.

\subsection{Relating the probability to be at an arbitrary vertex with the probability to be at the initial vertex}
The following \lcnamecref{eq:lemma} shows that it suffices to bound $P[0, t']$ for a time interval
$t'\in [t-w', t+w']$, as this would imply bounds on $P[w, t]$ for $w\leq w'$.

\begin{lemma} \label{eq:lemma}
Suppose that $n \geq 2$, $t \geq w'$ and $P[0,t'] \leq p_0$ for all $t' \in [t-w',t+w']$, where $w'<n/2$.
Then for all $w \in \{0,1,\ldots,w'\}$, we have \begin{equation}P[w,t] \leq \frac{n^w}{w!}  p_0.\end{equation}
\end{lemma}

\begin{proof} 
We prove the contrapositive: suppose that $P[w,t] = p_w$ for some $0<w<n/2$; then there exists a $t' \in [t-w,t+w]$ such that $P[0,t'] \geq \frac{w! p_w}{n^w}$.

We do this by showing two inequalities:
\begin{equation}
\label{eq:l11}
\max\left( |\alpha^{(t)}_{w, \leftarrow}|^2, |\alpha^{(t+1)}_{w-1, \rightarrow}|^2 \right) \geq \frac{w}{n-w} |\alpha^{(t)}_{w, \rightarrow}|^2 ,   
\end{equation}
\begin{equation}
\label{eq:l12}
P[w-1, t-1] 
\geq  |\alpha^{(t)}_{w, \leftarrow}|^2 .   
\end{equation}

These two inequalities imply that one of $P[w-1, t-1], P[w-1, t+1]$
is at least
\begin{equation}
\label{eq:prevstep}
\min\left( \frac{w}{n-w} |\alpha^{(t)}_{w, \rightarrow}|^2, 
 |\alpha^{(t)}_{w, \leftarrow}|^2 \right).\end{equation}
Because of \eqref{eq:pwt}, we must either have 
$|\alpha^{(t)}_{w, \rightarrow}|^2 \geq \frac{n-w}{n} P[w, t]$
or $|\alpha^{(t)}_{w, \leftarrow}|^2 \geq \frac{w}{n} P[w, t]$.
In both cases \eqref{eq:prevstep} is at least $\frac{w}{n} P[w, t]$.

By repeating this argument $w$ times, we get that $P[0, t']$, for some $t'\in[t-w, t+w]$ is at least
\begin{equation}
p_w \cdot \frac{w}{n} \cdot \frac{w-1}{n} \cdot \ldots \cdot \frac{1}{n} =  \frac{w! \, p_w}{n^w}.
\end{equation}

We now prove \eqref{eq:l11} and \eqref{eq:l12}.
To prove \eqref{eq:l11}, 
consider vertices $x\in\{0, 1\}^n$ for which $|x|=w$.
Before the coin flip $D_n$, the amplitudes of $\ket{x, i}$
with $x_i=0$ are equal to $\frac{\alpha^{(t)}_{w, \rightarrow}}{\sqrt{(n-w){\binom{n}{w}}}}$
and the amplitudes of $\ket{x, i}$
with $x_i=1$ are equal to $\frac{\alpha^{(t)}_{w, \leftarrow}}{\sqrt{w{\binom{n}{w}}}}$.
After applying the coin flip $D_n$, 
the amplitudes of $\ket{x, i}$
with $x_i=1$ become equal to
\begin{equation}
    \frac{2(n-w)}{n} \frac{\alpha^{(t)}_{w, \rightarrow}}{\sqrt{(n-w){\binom{n}{w}}}} - \frac{n-2w}{n} 
\frac{\alpha^{(t)}_{w, \leftarrow}}{\sqrt{w{\binom{n}{w}}}} .
\end{equation}
After the shift operation, each of those becomes an amplitude
of $\ket{y, i}$ with $|y|=w-1$ and $y_i=0$. 
Since $\ket{w-1, \rightarrow}$ consists of ${\binom{n}{w}-1} (n-w+1) = {\binom{n}{w}} w$ such $\ket{y, i}$,
we have
\begin{equation}
    \alpha^{(t+1)}_{w-1, \rightarrow} 
    =
\frac{2\sqrt{w(n-w)}}{n} \alpha^{(t)}_{w, \rightarrow} 
-
\frac{n-2w}{n} \alpha^{(t)}_{w, \leftarrow}.
\end{equation}
Assume that $|\alpha^{(t)}_{w, \leftarrow}| < \frac{\sqrt{w}}{\sqrt{n-w}} |\alpha^{(t)}_{w, \rightarrow}|$. (Otherwise, \eqref{eq:l11} is immediately true.)
Then,
\begin{equation}
\label{eq:afterq} 
 |\alpha^{(t+1)}_{w-1, \rightarrow}| \geq
\left( \frac{2\sqrt{w(n-w)}}{n} - \frac{\sqrt{w}}{\sqrt{n-w}}
\frac{n-2w}{n} \right) |\alpha^{(t)}_{w, \rightarrow}| .
\end{equation}
The equation \eqref{eq:l11} now  follows from 
\begin{equation}
    \frac{2\sqrt{w(n-w)}}{n} - \frac{\sqrt{w}}{\sqrt{n-w}}
\frac{n-2w}{n} 
= 
\frac{\sqrt{w}}{\sqrt{n-w}} .
\end{equation}

To prove \eqref{eq:l12}, we simply note that $S\ket{w-1, \rightarrow} = \ket{w, \leftarrow}$. Since coin flip $D_n$ moves the
amplitudes between $\ket{w-1, \leftarrow}$ and $\ket{w-1, \rightarrow}$, 
\eqref{eq:l12} expresses the fact
that all the
amplitude at $\ket{w, \leftarrow}$ after $t$ steps must have been
at $\ket{w-1, \leftarrow}$ or $\ket{w-1, \rightarrow}$ one step earlier.
This fact is obviously true.
\end{proof}

\subsection{Bounding the probability to be at the initial vertex}
Previously we demonstrated how the probability at a hypercube vertex is related to the probability at the initial vertex $P[0,t]$. Now we derive an explicit expression of $P[0,t]$ and apply it to  upper-bound   this probability.

In more details, we express (the square root of) the probability $P[0,t]$ through  Chebyshev polynomials, see \eqref{eq:3.28}; then we apply  an integral representation of Chebyshev polynomials (\cref{thm:bessel-chebyshev}), arriving at an integral representation of the  probability $P[0,t]$ (\cref{eq:thmP0t}).

The latter representation involving the integral $\int_0^{\infty} x^{-1} J_{t}(x)  \cos(x  z) \dx$  turns out far more suitable (compared to the expression involving Chebyshev's polynomials) for proving an upper bound.

This is in part due to a clear separation between the oscillating part ($J_t$) and the exponential part ($\cos^n (x/n)$). This distinction helps to explain how the ratio $t / n$ leads to the oscillating behavior of $P[0,t]$ observed in the lower   panel of \cref{fig:figure1_AD} (with $n=50$). To illustrate this, examine the interplay between both parts in  \cref{fig:3.2}, for $n=50$ and different values of $t$. 
\begin{itemize}
    \item When $t/n$ is ``small'' (e.g., see the upper row in \cref{fig:3.2} with $n=50$, $t=0.12n$), the Bessel function's first maximum largely overlaps with the contribution of $\cos^n(\cdot)$ in the vicinity of $x=0$. However, the subsequent maxima of $\cos^n(\cdot)$ occur where the oscillations of $J_t$ cancel out, resulting in an insignificant contribution to the integral. Therefore, the value of the integral is dominated by the values of the integrand near $x=0$ and the probability $P[0,t]$ is large.
    \item Now consider the case when  $t/n$ is ``large'' (see the middle row in \cref{fig:3.2} for the case $n=50$, $t=3.12n$). While the Bessel function's values near $x=0$ are negligible, the following local maximum greatly overlaps with the next peak of $\cos^n(x/n) $ (in the vicinity of $x\approx \pi n$). This results in a large contribution to the integral, which is not canceled out by the following peaks of  $\cos^n(x/n) $ (where the Bessel function's fluctuations make the contribution to the integral negligible).
    \item Finally, when  $t/n$ is ``just right'' (see the lower row in \cref{fig:3.2} for the case $n=50$, $t=0.84n$), the first maxima of the Bessel function occur when $\cos^n(x/n)$ is near zero. This ensures that the resulting integrand's oscillations near each peak of  $\cos^n(x/n) $ largely cancel out. We proceed to quantify the remaining integral's value and show that it indeed is exponentially small in $n$.
\end{itemize}

\begin{figure}[H]
		\centering
		\includegraphics[width=\textwidth,height=\textheight,keepaspectratio]{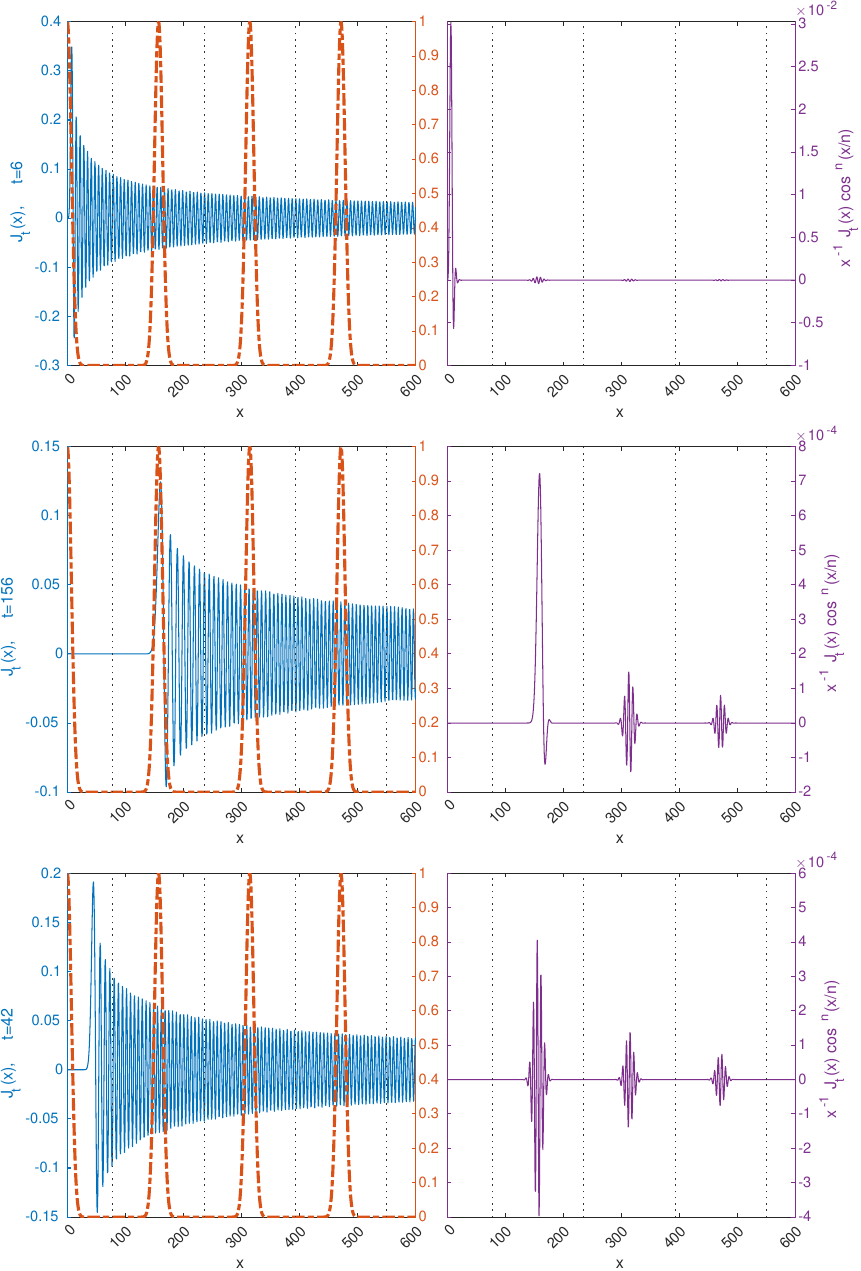}
		\caption{Illustration of the integrand and its factors for $n=50$ and 
		$t=0.12n$ (upper row),  
		$t=3.12n$ (middle row),  
		$t=0.84n$ (lower row).
	Panels on the left show $J_t(x)$, solid line, with units on the left axis and $\cos^n(x/n)$, dash-dotted line, with units on the right axis. {Panels on the right} depict the integrand $x^{-1} J_{t}(x)  \cos^n \frac{x}{n} $. The dotted gridlines divide the real line into subintervals separating the peaks of $\cos^n(x/n)$.}
		\label{fig:3.2}
\end{figure}

Despite this intuition, at first glance both expressions of $\sqrt{P[0, t]}$ (the discrete sum involving Chebyshev polynomials and the integral)   suffer the same drawback (or the advantage) of catastrophic cancellations. Indeed, while the integrand's extreme values near each peak of $\cos^n(x/n) $ are polynomially decaying in $n$, cancellations ensure that the overall contribution is exponentially small! This property prohibits from making simple estimates of $\sqrt{P[0, t]}$  in either case. However, here the main advantage of the integral representation manifests itself: instead of computing integrals over real intervals, we can exploit  Cauchy's integral theorem to transform them to contour integrals in the complex plane. Since the integrand turns out to be non-oscillatory in the complex plane in directions orthogonal to the real axis, this provides an efficient way to estimate the value of $P[0,t]$. This approach leads to the asymptotic estimates in \cref{thm:L11} and the desired upper bound on the probability $P[0,t]$.

The rigorous arguments begin with a proof of 
the following representation of Chebyshev polynomials of even order, valid for $|z| \leq 1$:
\begin{lemma}\label{thm:bessel-chebyshev}
Let $t$ be a positive even integer; then the following equality is valid for all $z \in [-1,1]$:
	\begin{align} \label{eq:TtoJ}
		T_{t}(z) = (-1)^{t/2} \, t \int_0^{\infty} x^{-1} J_{t}(x)  \cos(x  z) \dx, 
	\end{align}
	where $T_{t}$ is the $t^{\text{th}}$-degree Chebyshev polynomial of the first kind and $ J_{t}$ stands for the Bessel function of the first kind.
\end{lemma}
\begin{proof}
	We start with a special case of the Sonine-Schafheitlin formula\cite[p.405, \S 13.42, Eq. (3)]{Watson1922}
	\begin{equation}\label{eq:l2e01}
	    \int_{0}^{\infty} \frac{J_{\mu}(ax)  \cos (bx)}{x} \dx
	=
	\mu^{-1} \cos \left(  \mu \arcsin (b/a) \right), \quad 0 < b \leq a.
	\end{equation}
	Taking $a=1$, $ b = z \in (0,1]$ and $\mu=t$ in \eqref{eq:l2e01} gives
	\begin{align}
		\int_0^{\infty} \frac{J_{t}(x) \cos (x \, z)}{x} \dx & = \frac{\cos\left( t  \cdot  \arcsin  z \right) }{t} .
	\end{align}
	It remains to notice that
	\begin{equation}
	    \cos\left( t  \cdot  \arcsin  z \right)
	=
	\cos\lr{t \lr{\frac{\pi}{2}  -  \arccos z}} =(-1)^{t/2} \cos (t \arccos{z}) 
	=(-1)^{t/2} T_{t}(z).
	\end{equation}
	Continuity and evenness of $\cos$ and $T_{t}$ ensures the validity of \eqref{eq:TtoJ} for $z \in [-1,0]$.
\end{proof}

The following \lcnamecref{thm:bessel-chebyshev} allows us to characterize the probability $P[0, t]$ to obtain the vertex $0^n$ after $t$ steps as follows: 
\begin{lemma}\label{thm:P0t}
Let $t$ be a positive integer; if $t$ is odd, then $P[0,t]=0$. If $t$ is even, then
\begin{equation}\label{eq:thmP0t}
    \sqrt{P[0,t]}  = \frac{1}{2^n} \sum_{m=0}^n \binom{n}{m} T_t\lr{1-\frac{2m}{n}} 
= t \lrv{ \int_0^\infty x^{-1} J_t(x)\cos^n \frac{x}{n} \dx  } \, .
\end{equation}
\end{lemma}
\begin{proof}
Due to the symmetry of the walk (w.r.t.\ permuting the coordinates of the hypercube), the amplitudes of $\ket{0^n, j}$ for all $j\in [n]$ remain equal after any number of steps. Therefore,
\begin{equation}
     P[0, t] = \lbra \psi_{start} | W^t | \psi_{start} \rket^2 .
\end{equation}
Let $\ket{ v_j}$ be the eigenvectors of $W$, with eigenvalues $\e^{\imagI \lambda_j}$. We can represent $\ket{\psi_{start}}$ as a linear combination of the
eigenvectors, $\ket{\psi_{start}} = \sum_j a_j \ket{ v_j}$, then 
\begin{equation}
    \sqrt{P[0, t]} = \lbra \psi_{start} | W^t | \psi_{start} \rket = \sum_j |a_j|^2 \e^{\imagI t \lambda _j} .
\end{equation}
As described in \cite{moore2002quantum}, $W$ has $2^{n+1}$ eigenvectors that have non-zero overlap with the starting state.
These eigenvectors can be indexed by $k\in\{0, 1\}^n$. 
For each $k\in\{0, 1\}^n$, we have a pair of eigenvectors 
$\ket{v^{+}_k}$ and $\ket{v^{-}_k}$  
with eigenvalues $\e^{\pm \imagI w_m}$ where $m=|k|$ is the Hamming weight of $k$ and $w_m$ is such that  $\cos w_m = 1 - {2m}/{n}$. Furthermore, all of those eigenvectors have equal overlap with the starting state: 
$ |\lbra v^{+}_k \ket{\psi_{start}}|^2 = |\lbra v^{-}_k \ket{\psi_{start}}|^2 = \frac{1}{2^{n+1}}$.
Therefore, we get
\begin{align}
\sqrt{P[0, t]} 
&= \frac{1}{2^{n+1}} \sum_{k\in \{0, 1\}^n}  \lr{ \e^{ \imagI t w_m} + \e^{- \imagI  t w_m} }
 \\
&= \frac{1}{2^n} \sum_{k\in \{0, 1\}^n} 
\cos w_{m}t \\
 &= \frac{1}{2^n} \sum_{m=0}^n \binom{n}{m}\cos w_mt \\
&= \frac{1}{2^n} \sum_{m=0}^n \binom{n}{m}\cos\left(t \cdot \arccos \left(1-2\cdot \frac m n\right)\right) \\
&= \frac{1}{2^n} \sum_{m=0}^n \binom{n}{m} T_t\left(1-2\cdot \frac m n\right). \label{eq:3.26}
\end{align}

 If $t$ is odd, then so is $T_t$ and  the terms in \eqref{eq:3.26} cancel out; otherwise,
using \eqref{eq:TtoJ} with $z=(n-2m)/n$ for  even integer $t$, 
\begin{align}
 \frac{1}{2^n}  \sum_{m=0}^{n} \binom{n}{m}
 T_t \left (1- \frac{2 m}{n}\right)  
& = (-1)^{t/2} t \int_0^{\infty} x^{-1} J_{t}(x)  \sum_{m=0}^{n}   \frac{1}{2^n} \binom{n}{m} \cos \left (x -  \frac{2x \, m}{n} \right) \dx
\\
& = (-1)^{t/2} t \int_0^{\infty} x^{-1} J_{t}(x)  \cos^n \frac{x}{n}  \dx \label{eq:3.28}
\end{align}
with the last step following from the binomial expansion of $(\e^{\imagI x/n}+\e^{-\imagI x/n})^n$.
\end{proof}

The most technical contribution of this manuscript is an upper bound of the integral appearing in \cref{thm:P0t}, provided that $t/n$ is suitably bounded. In what follows, we use $\nu\in \mathbb R$  instead of $t$ (since the latter was assumed  to be an even integer in our notation before).

\begin{restatable}{theorem}{thesecondthm}\label{thm:L11}
Denote $ a_k =  (k-0.5)\pi $, $ k=1,2,3,\ldots $; let a positive constant $\ratioLower\in (\frac{\pi}6,1)$ be fixed.

	For $ n \geq 2 $ and $\nu>1$, satisfying $ \nu \in (n \ratioLower, n ) \subset  \lr{\frac{ \pi n}{6} ,  n}$,  the following estimates hold:
	\begin{align}
	\textbf{The tail:}\qquad
	& \lrv{ \int_{n a_n}^\infty x^{-1}  J_\nu (x)   \cos^n\lr{ \frac{x}{n}  } \dx   }   < \frac{100    \sqrt{n} }{2^n    }; \label{eq:L11.01}\\
	\textbf{The middle part:}\qquad
	& 	\lrv{ \int_{n a_1}^{na_n} x^{-1}  J_\nu (x)   \cos^n\lr{ \frac{x}{n}  } \dx   }   < \frac{4000 \sqrt n}{1.541^{n} };  \label{eq:L11.02}\\
	\textbf{The bulk:}\qquad
	& \lrv{ \int_{0}^{na_1} x^{-1}  J_\nu (x)   \cos^n\lr{ \frac{x}{n}  } \dx     }  	<  \frac{3}{   (1+\ratioLower)^{0.5\ratioLower n } }. \label{eq:L11.03}
	\end{align}
\end{restatable}
The proof of the \lcnamecref{thm:L11} is deferred to \cref{sec:appendix}; below we roughly  sketch a  reasoning behind the proof.

As   \cref{fig:3.2} might suggest, it is advantageous to divide the whole integral into a sum of integrals over subintervals of the form ${[(k-0.5)\pi n , (k+0.5)\pi n]}$, separating the individual peaks of $\cos^n(x/n)$ and estimating the contribution of each integral separately. It turns out that there are three regimes to consider: 
\begin{enumerate}
  \item When $k \geq n$, the integral over each subinterval  $[(k-0.5)\pi n , (k+0.5)\pi n]$ turns out to be of order $2^{-n} n^{1/2} k^{-5/2}$; the series (taking the sum over $k : k \geq n$) converges, leading to the tail estimate \eqref{eq:L11.01}.
  \item When $1 \leq k < n$, we estimate the integral over each subinterval  $[(k-0.5)\pi n , (k+0.5)\pi n]$ as  $1.541^{-n} n^{-1/2} $, with the constant $1.541$ coming from an upper bound on a certain auxiliary function. Since there are $n-1$ such subintervals, this leads to the  estimate of the middle part \eqref{eq:L11.02}.
  \item In both previous cases, the integral is going to be estimated via an excursion into the complex plane. However, for the remaining interval $(0, 0.5 \pi n)$ this is not a viable option anymore. Instead, here we combine direct bounds on the Bessel function $J_\nu$ and the  $\cos^n(x/n)$ term when $x$ is ``small''. These estimates (in particular, the estimate of the integral over the interval $[0,\nu]$)  lead to the dominating contribution (``the bulk'') \eqref{eq:L11.03}. 
\end{enumerate}

To estimate the integrals of the form $\int_{n a_k}^{na_{k+1}} x^{-1}  J_\nu (x)   \cos^n\lr{ \frac{x}{n}  } \dx $, we form a rectangle in the complex plane, consisting of the real interval ${[n a_k, na_{k+1}]}$, a parallel interval ${[n a_k + \imagI R, na_{k+1} + \imagI R]}$ for some $R>0$, as well as two segments orthogonal to the real axis forming a closed contour. The Cauchy integral theorem then implies  that the integral over the real integral coincides with the integral over the rest of the contour. While the integrand still is oscillatory  over the set ${[n a_k + \imagI R, na_{k+1} + \imagI R]}$, its absolute value vanishes as $R \to +\infty$. Thus, taking the limit $R\to +\infty$, we are left with integrals over the rays of the form $ [na_k, na_k + \imagI \infty) $.

For this strategy to work, we need the integrand to be holomorphic (in a domain containing the integration contour). We note that the integrand $  x^{-1}  J_\nu (x)   \cos^n\lr{ \frac{x}{n}  } $ is the real part of
the function
\begin{equation}\label{eq:f_def}
    f(z)  =   \frac{\hankel(z) \cos^n(z/n) }{z}, \quad  z\in \mbb C,\ \lrv z > 1,
\end{equation}
 where  $ \hankel $ is  the Hankel function of the first kind\footnote{By slightly abusing terminology, we will refer to $ \hankel $ as  \textit{the} Hankel function.} (and satisfies $\Re  \hankel(x) =  J_\nu(x) $ for real $x$). Now $ f $ is holomorphic in the necessary domain, and we make use of the estimate
 \begin{equation}
     \lrv{\int_{n a_k}^{na_{k+1}} x^{-1}  J_\nu (x)   \cos^n\lr{ \frac{x}{n}  } \dx }
 =
 \lrv{
  \Re \int_{n a_k} ^{n a_{k+1}}
f(x) \dx 
 }
 \leq  \lrv{
 \int_{n a_k} ^{n a_{k+1}}
 f(x) \dx 
 }.
 \end{equation}
 Finally, to bound the values of $f$ on the rays $ [na_k, na_k + \imagI \infty) $, we  will employ asymptotic expansions of the Hankel function $\hankel(z)$. The case of large $z$ is straightforward \cite[\href{https://dlmf.nist.gov/10.17.iv}{\S 10.17(iv)}]{DLMF}. However,  when the argument $z$ and order $\nu$ are of similar magnitude, there is no suitable expansion of the Hankel function with error bounds (to the best of our knowledge). This leads to the most involved argument in this paper, where we relate the Hankel function to the modified Bessel function of the second kind, then  consider its asymptotic expansion and bound the relevant  error term, as detailed in \cref{sec:4.0}.

\subsection{Concluding the proof}
Finally, we are ready to show \cref{thm:main}. As explained at the beginning of \cref{sec:proof},  the overall strategy for bounding $P(x,t)$  is to consider two cases, depending on  the Hamming weight of  the vertex $x$. Now we sketch an informal outline of the remaining argument.

Let  $c<  0.5$ be a constant (its value to be determined); consider two cases: when $ \lrv x  $ is between $cn$ and $(1-c)n$ (then the binomial coefficient $\binom{n}{\lrv x}$ is ``large'') and when $\lrv x < cn$ (then the vertex $x$ is ``close'' to the initial vertex). We also set $t=  (1-c)n$ (thus the case $\lrv x > (1-c)n$ need not be considered) and $\alpha=(1-2c)$ (the constant appearing in \cref{thm:L11}).
\begin{enumerate}
    \item In the former case, we apply    $P(x, t)\leq 1/{\binom{n}{\lrv x}} \leq  1/{\binom{n}{{cn}}}$. The latter quantity is (roughly; for a precise statement, see \eqref{eq:3.44}) bounded by 
    \begin{equation}\label{eq:3.34}
    2^{-H(c) n}   
				= 
\lr{ c^{c} (1-c)^{(1-c)} }^{n}.
\end{equation}

    \item In the latter case, we rely on the already proven relationship between $P(x,t)$ and $P(0^n,t')$ for some $t' \in [t-cn, t+cn] \subset [(1-2c)n, n] $; specifically, we apply \cref{eq:lemma,thm:P0t} and intend to upper-bound the integral via \eqref{eq:L11.03} (which is the dominating term)  with $\alpha=(1-2c)$. This (again, roughly; the precise statement is \eqref{eq:3.45}) leads to the following upper bound on $P(x,t)$: 
    \begin{equation}\label{eq:3.35}
     \lr{ \frac{\e (1-c)^{1-c}}{ (2-2c)^{(1-2c)}} }^n
 =
 \lr{ \frac{\e (1-c)^{c}}{2^{(1-2c)}} }^n.
\end{equation}
\end{enumerate}
Now by  balancing the estimates \eqref{eq:3.34} and \eqref{eq:3.35}, we obtain the equilibrium value $c = 0.133682\ldots $. However, the above reasoning is faulty: it ignores the fact  that
\cref{thm:L11} requires
$t'$ to be within the \emph{open} set $(\alpha n, n)$. Also, $1/{\binom{n}{{cn}}}$ is, in fact, not upper-bounded by $2^{-H(c) n}$  (a polynomial factor  is missing). To fix these issues, we choose slightly smaller values of $\alpha $ and $t$, leading to a slightly worse overall estimate.

This leads to the following proof of 
\cref{thm:main} (restated here for convenience).
\themainthm*
\begin{proof}
We   assume that $n$ is sufficiently large, i.e., $\lr{0.86628n,0.86632n} \cap \mbb  Z \neq \emptyset$ and   $n > \max\{ n_0,1 \} $, where $n_0$ is described below.
 Let  
$t$ be from the indicated range; set $w' = \constC n $ and 
	$\ratioLower=\constA$ in \cref{thm:L11}. The chosen constants ensure that $[t-w', t+w']  \subset   (\constA n ,n)  $.

	Then $(1+\ratioLower)^{0.5\ratioLower} > 1.22302 $ and \cref{thm:L11} gives the estimate (for $t' \in (\alpha n, n)$)
\begin{equation}
    \lrv{ \int_{0}^\infty x^{-1}  J_{t'} (x)   \cos^n\lr{ \frac{x}{n}  } \dx }
<
\frac{100     \sqrt{n} }{2^n    }
+
\frac{4000 \sqrt n}{1.541^{n} }
+
\frac{3}{   1.22302^{n} }
=
O \lr{    1.22302^{-n }  }.
\end{equation}
This together with \cref{thm:P0t} implies that there are positive constants $C_0, C_1$ and a positive integer $n_0$ such that  for all $n>n_0$ and all $t' \in (\alpha n, n) = (\constA n ,n) $ the probability $P[0,t']$ can be estimated as
    \begin{equation}
        P[0,t']
    <  \frac{C_0 n {t'}^2} {   1.22302^{2n} } 
    \leq \frac{ C_0 n^3}{1.49578\ldots^n}
    <  \frac{C_1}{1.49578^n}, \quad 
    t' \in (\constA n ,n) .
    \end{equation}
In particular,
\begin{equation}
    P[0,t']
 <  \frac{C_1}{1.49578^n} \quad \text{for all } t' \in [t-w', t+w']  
 \subset   (\constA n ,n).
\end{equation}

Fix any $x \in \{0,1\}^n$ and let $w = |x|$. 
		If $w >n- w'>   t$, then the discrete-time quantum walk cannot reach $x$ in $t$ steps, thus  $P(x,t)=0$. Therefore there are two possibilities to consider:
    \begin{enumerate}
        \item $w \leq w'$, then \cref{eq:symm} and \cref{eq:lemma} with $p_0 = C_1 \cdot 1.49578^{-n} $ give
        \begin{equation}
            P(x,t) = \frac{P[w,t]}{\binom n w } 
			\leq p_0  \frac{n^w  (n-w)! }{ n!}
			<
			 \frac{C_1}{1.49578^n}  \cdot n^{n c}  \cdot \frac{(n-nc)!}{n!}
        \end{equation} 
		with $c:= w'/n = \constC$.
			{Lower and upper bounds} \cite{robbins1955} on the factorials  yield
			\begin{equation}
			    n!  \geq   \lr{ \frac{n}{\e } }^n   \e^{\frac 1 {12n+1}} \sqrt{ 2\pi n}  
			 \, \text{ and } \, 
			 (n-nc)! \leq \lr{ \frac{n(1-c)}{\e } }^{n(1-c)}    \e^{\frac 1 {12n(1-c)}} \sqrt{ 2\pi n(1-c)}  ,
			\end{equation}
			leading to
			\begin{equation}
			    \frac{ (n-nc)! }{n!} \leq  
				\lr{ \frac{ n ^{(1-c)} (1-c)^{(1-c)} \cdot \e}{\e^{(1-c)}  \cdot n} }^{n}  \e^{\frac 1 {12n(1-c)} - \frac 1 {12n+1}}
				\sqrt{ (1-c)}.
			\end{equation}
			Now, since $  \e^{\frac 1 {12n(1-c)} - \frac 1 {12n+1}} < \e^{\frac 1 {12(1-c)}}< 2$ and $\sqrt{ (1-c)} < 1$, 
				we have
				\begin{equation}
				    \frac{ (n-nc)! }{n!}  
				< 
				2n^{-nc} \e^{nc} (1-c)^{(1-c)n} 
				< 2\cdot  0.99068^{-n} n^{-nc},
				\end{equation}
		because  $\e^c (1-c)^{(1-c)} = 0.990681\ldots^{-1} < 0.99068^{-1}$.
				 Let $C_2= 2C_1$, then we conclude
				 \begin{equation}\label{eq:3.45}
				     P(x,t) < 
				\frac{C_2}{(1.49578 \cdot 0.99068)^n   }
				<
				\frac{C_2}{\constFinal^n}
				=O\lr{\constFinal^{-n}}.
				 \end{equation}

        \item $ w' < w  \leq  n-w' $, then
        \eqref{eq:symm}
        and a 
        lower bound \cite[Lemma 9.2]{mitzenmacher_upfal_2005}  on the binomial coefficient in terms of the binary entropy function $H$ gives
        \begin{equation}\label{eq:3.44}
             P(x,t) 
        \leq  \binom{n}{w'}^{-1} 
        \leq (n+1)2^{-n H(w'/n)} =
        (n+1) \lr{ 2^{H(\constC)} }^{-n}.
        \end{equation}
        Since $ 2^{H(\constC)} =1.48189\ldots$, we arrive at 
        $  P(x,t)  =
        (n+1) \cdot 1.48189\ldots^{-n}
        = O(\constFinal^{-n}) $.
    \end{enumerate} 
\end{proof}

\section{Summary and outlook}

We have shown that quantum walk on the hypercube quickly disperses over vertices so well that no 
vertex has more than an exponentially small part of the quantum state on it. This dispersion property is 
significantly stronger than the standard mixing property which requires that the walk has spread almost
uniformly over the vertices but allows significant spikes on particular vertices.

Our computer simulations show that, after $O(n)$ steps of the standard discrete time quantum walk on a $n$-dimensional hypercube, the probability of being at any vertex of the hypercube is at most $1.93...^{-n}$. 
Since the $n$-dimensional hypercube has $2^n$ vertices, this dispersion is close to the maximum possible. While there is a number of results about fast mixing of quantum walks~\cite{levin2017markov,randall2006rapidly}, such strong dispersion results have been rare.
Rigorously, we can show that the probability of the walker being at any vertex is $1.4818...^{-n}$. The proof uses an intricate argument about asymptotics of Bessel and Hankel functions. 

All of those results are for a starting state where the walker is localized in one vertex of the hypercube and the initial direction for the walker is the uniform superposition of all $n$ possible directions. 
For the case when the initial direction of the walker is a basis state corresponding to one direction, 
the quantum walk shows an oscillatory localization, with a large fraction of the walker's state staying either in the starting vertex or its neighbouring vertex, depending on the step of the walk \cite{OscillatoryLoc}.

A particular application of our results would be to show an exponential advantage for quantum algorithms in the query model for the case when the main non-query transformation is a quantum walk, along the lines of \cite{brandao2013exponential,aaronson2018forrelation}. 
A technical difficulty here is that the dispersion result only holds for starting states where the direction register of the walker is in the uniform superposition. Thus, proving an exponential advantage requires generalizing the conditions from \cite{brandao2013exponential} for achieving an exponential advantage, which is a subject of future work.

\section*{Acknowledgements}
We thank Ashley Montanaro for suggesting the problem and the motivation for studying it and Raqueline Santos for participating in the early work on this subject.
This work has been supported by Latvian Council of Science (project no.\ lzp-2018/1-0173).

\appendix

\section{Bounding the integral}\label{sec:appendix}

\subsection{Preliminaries}
Throughout the proof, we will make use of standard notation of some special functions; $ \Gamma $ will stand for the gamma function;  $ B$ denotes  the beta function,  satisfying  $ B(x,y) = \Gamma(x)\Gamma(y)/\Gamma (x,y)$. The Bessel functions of the first  and second kind will be denoted as  $ J_\nu(x) $ and $ Y_\nu(x) $, respectively;  and the Hankel function of the first kind  is denoted as $ \hankel (x)  =  J_\nu(x)  +  \imagI Y_\nu(x) $. As in the statement of the \lcnamecref{thm:L11}, 
we denote $ a_k =  (k-0.5)\pi $, $ k=1,2,3,\ldots $, and let a positive constant $\ratioLower\in (\frac{\pi}6,1)$ be fixed. Furthermore, let $ \nu >0 $ and suppose an integer $ n \geq 2  $
satisfies
\[
n \geq  \frac{1}\ratioLower
\quad\text{and}\quad 
\ratioLower < \frac{\nu}{n} < 1.
\]
Also, for $ a,b \in \mbb R $, with $ a<b $,  we define $ D_{a,b}  =  \lrb{  x+ \imagI y \ \vline\   x \in [a,b],  \  0\leq y < \infty  }  $ (i.e.,  $ D_{a,b} $ stands for an  ``infinite rectangle'' in the complex  plane with one side being the real interval $ [a;b] $).

\subsection{Asymptotic expansions of the Hankel function}
\paragraph{The case of large argument.}
	When $ z \in \mbb C $,  $ \lrv{\arg z} \in [ 0, \pi - \delta) $ for an arbitrary small  positive constant $ \delta $,
we have
\cite[\href{https://dlmf.nist.gov/10.17.iv}{\S 10.17(iv)}]{DLMF}
\begin{equation}\label{eq:1.2.01}
	\hankel(z) = \sqrt{\frac{2}{\pi z}} \,  \e^{\imagI (z - \nu \pi/2 - \pi/4)} \lr{  1 + \rho(\nu, z) },
\end{equation}
where the principal branch of $ \sqrt z $ is used,
and
\begin{equation}\label{eq:1.2.02}
	\lrv{\rho(\nu,z) }  \leq    \lr{ \nu^2 -1/4}  \lrv z^{-1} \exp\lr{ \lr{ \nu^2 -1/4} \lrv z^{-1}  }.
\end{equation}

For future reference, we note that whenever $z$ additionally satisfies $\lrv z \geq n^2$, we have (since $\nu   <n$)
\begin{equation}\label{eq:1.2.03}
	\lrv{\rho(\nu,z) }  <     \nu^2    \lrv z^{-1} \exp\lr{ \nu^2   \lrv z^{-1}  }  < \e 
\end{equation}
and
\begin{equation}\label{eq:1.2.04}
	\lrv{\hankel(z)} \leq  \sqrt{\frac{2}{\pi }}  \lr{  1 + \e  } \lrv { \e^{\imagI z }} \lrv z^{-1/2} <  \frac{3 \e^{-\Im(z)}}{\lrv z^{1/2}}.
\end{equation}

\paragraph{The case when argument and order are of similar magnitude.}
	\begin{lemma}\label{thm:L01}
	When  $ \arg w \in [0; \pi /2) $ and $ \Re  w > 1  $, we have
	\begin{equation}\label{eq:L01e01}
		\hankel(\nu w) 
		=
		-\imagI \,  \sqrt{ \frac{2}{ \pi \nu } } \,  \frac{\e^{ \imagI \nu (   \sqrt{w^2- 1}  - \arccos (1/w)    ) }}{(1-w^2)^{0.25}} \lr{ 1+ \eta(\nu,-\imagI w)},
	\end{equation}
	where the fractional powers and the logarithm take their principal values on the positive real axis 
	and the term $ \eta  $ can be bounded as
	\[ 
	\lrv{ \eta(\nu, - \imagI w)}   \leq \exp \lr{ \frac{2 \mc V_{+\infty, - \imagI w}}{ \nu}  } \, \frac{2 \mc V_{+\infty, - \imagI w}}{ \nu}  ,
	\]
	and the quantity $  \mc V_{+\infty, - \imagI w}$ satisfies
	\[ 
	\mc V_{+\infty, - \imagI w}  \leq \frac{1}{12} + \frac{1}{6 \sqrt 5} + \lr{\frac{4}{27}}^{1/4} + \frac{c^2(c^2+2)}{\sqrt 8 (c^2-1)^{2.5}}, 
	\qquad   c:= \Re w > 1, \  \Im w \geq 0.
	\]
	
	Moreover, when  $ \nu \geq 1 $ and  $ w \in \mbb C $ satisfies $ \Re(w) \geq \pi /2 $ and  $   \Im(w) \geq 0 $, we have 
	\begin{equation}\label{eq:L01e02}
		\lrv{ 1+ \eta(\nu,-\imagI w)} \leq 430.
	\end{equation}
\end{lemma}
The proof is postponed until \cref{sec:4.0}.

\subsection{Auxiliary lemmata}
Here we list a few somewhat disjoint auxiliary results that will be useful in the subsequent analysis. 
\begin{lemma}\label{thm:L02}
	Suppose that $ f : D \to \mbb C $ is holomorphic in $ D $, where $ D $ is a domain containing the region $D_{a,b} :=   \lrb{  x+ \imagI y \ \vline\   x \in [a,b],  \  0\leq y < \infty  } $, for some reals  $ a<b $.  Moreover, assume that 
	\begin{enumerate}
		\item $ \lim\limits_{R \to +\infty}  \sup_{x \in [a,b]}  \lrv{ f(x + \imagI R)}   = 0 $, and
		\item integrals $ \int_0^\infty  f(a + \imagI y) \dy    $,  $ \int_0^\infty  f(b + \imagI y) \dy  $ converge.
	\end{enumerate}
	Then 
	\begin{equation}\label{eq:L02e01}
		\int_a^b f(x) \dx 
		=
		\imagI  \int_0^ \infty \lr{f(a+\imagI y)  -  f(b+\imagI y) } \dy  .
	\end{equation}
	
\end{lemma}
\begin{proof}
	For every $ R>0 $ consider the positively oriented rectifiable curve $ \gamma_R $ consisting of the line segments
	\[ 
	[a,b] \cup \lrb{  b+ \imagI y \ \vline\  y \in [0,R] }  \cup \lrb{  x + \imagI R \ \vline\  x \in [a,b] }  \cup \lrb{  a + \imagI y \ \vline\  y \in [0,R] }.
	\]
	By Cauchy's integral theorem   we have $ \oint_{\gamma_R} f(z) \dz = 0  $, i.e.,
	\[ 
	\int_a^b f(x) \dx 
	+
	\imagI \int_0^R f(b+\imagI y) \dy 
	-
	\int_a^b f(x+\imagI R) \dx 
	-
	\imagI \int_0^R f(a+\imagI y) \dy 
	=0.
	\]
	Rearrange this equality as
	\[ 
	\int_a^b f(x) \dx 
	=
	\imagI \lr{\int_0^R f(a+\imagI y) \dy   -    \int_0^R f(b+\imagI y) \dy }  + \int_a^b f(x+\imagI R) \dx 
	\]
	and take the $ R \to \infty $ limit. Since 
	\[ 
	\lrv{\int_a^b f(x+\imagI R) \dx }   \leq  (b-a)   \sup_{x \in [a,b]}  \lrv{ f(x + \imagI R)},     
	\]
	which tends to 0 by the assumptions of $f$,
	we are done.
\end{proof}

\begin{lemma}\label{thm:L03}
	For all $ a>0 $ the following equalities hold:
	\begin{align}
		& \int_{0}^{\infty} \frac{\dy }{(a^2 + y^2)^{3/4}} =   \frac{B(0.5, 0.25) } {  2\sqrt{a}} \approx 2.62206\ldots \cdot a^{-0.5},    \label{eq:L03.01}\\
		& \int_{0}^{\infty} \frac{\dy }{(a^2 + y^2)^{5/4}} =   \frac{B(0.5, 0.75)}{2 \sqrt{a^3}} \approx 1.1984\ldots \cdot a^{-1.5}.   \label{eq:L03.02}
	\end{align}
\end{lemma}
\begin{proof}
	The first equality can be proven as follows: substitute $ t=a^2y^2 $, which gives $ \dy  = 0.5 a  t^{-1/2}  \, \dt $. Then
	\[ 
	\int_{0}^{\infty} \frac{\dy }{(a^2 + y^2)^{3/4}}
	=
	\frac{a}{2} \cdot \frac{1}{a^{3/2}}  \int_{0}^{\infty} \frac{ t^{-1/2} \dt }{(1+t)^{3/4}}
	=
	\frac{1}{2 a^{1/2}}\int_{0}^{\infty} \frac{ t^{1/2-1} \dt }{(1+t)^{1/2 + 1/4}}.
	\]
	Now from the identity  $ B(x,y) = \int_{0}^{\infty} \frac{ t^{x-1} \dt }{(1+t)^{x+y}} $  we recognize  the integral  on the RHS as the beta function value  $B(0.5,0.25)  = \Gamma(0.5) \Gamma(0.25) / \Gamma(1.25)  \approx 5.24412\ldots $.  
	In a similar manner one shows the second equality.
\end{proof}

\begin{lemma}\label{thm:L04}
	For all   $ t \in  [0; \pi /2) $ the inequality $ \cos t \leq  \e^{-t^2/2} $ holds.
\end{lemma}
\begin{proof}
	Let $ h(t) = \ln (\cos(t))  + t^2/2$ and consider $ h'(t) = t-\tan(t) $. Since  $ h'(0)=0 $ and $ h'(t)<0 $ for $ t\in (0;\pi/2) $, $ h $ attains its maximum at $ t=0 $, i.e., $ \ln(\cos t)  \leq -t^2/2$ with equality at $ t=0 $. Exponentiating gives the desired result.
\end{proof}

\begin{lemma}\label{thm:L06}
	
	Let  $ g $ be defined in the   set $ \lrb{z \in \mbb C \ \vline\   \Re z  \geq 1,   \Im z \geq 0 } $ via 
	\[
	g(z) = z - \sqrt{z^2-1} + \arccos(1/z),
	\]
	where the square root and the inverse cosine functions take their principal values on the positive real axis.
	Then $ \Im (g(z))  < 0.2607$ for all $ z $ with $ \Re z  \geq 1 $,   $ \Im z \geq 0  $.
\end{lemma}
\begin{proof}
	Let $ R>0 $; we will assume that $ R $ is large enough, e.g., $ R>1 $.
	
	Since $ g $ is holomorphic on the domain $  \lrb{  x + \imagI y \ \vline\  x \in (1,R),\  y \in (0,R) }   $ and continuous on its closure, its imaginary part $ \Im(g(z)) $ is harmonic and attains its maximum on the boundary of this region. Moreover, as    $ g(z ) \in \mbb R $ when $ \Im z = 0 $ and
	\[ 
	\lim\limits_{\lrv z \to \infty}
	\lrv{ z - \sqrt{z^2-1}} 
	= 
	0
	\quad\text{and}\quad
	\lim\limits_{\lrv z \to \infty}
	\arccos(1/z)
	= 
	\pi /2,
	\]
	it follows that 
	$ \Im(g(z)) $  must  attain its  maximum   on  the segment with $ \Re z =1 $, $ \Im z \in [0,R] $. 
	
	Let us separate the real and imaginary part of $ g $, i.e., introduce real-valued bivariate functions $ u,v $ satisfying $ g(x+\imagI y)  = u(x,y) + \imagI v(x,y)  $. By the arguments above, we need to show that $ v(1,y)  < 0.2607$ for all $ y \in [0,R] $, for arbitrarily large $ R $. We will show that (for $ R>1 $) the function  $ v(1,\cdot) $ has a single local maximum at $ y_0    \approx 0.86883\ldots $ where its value is less than $  0.2607 $. To calculate  the derivative of $y \mapsto  v(1,y)   $  and show that it is positive  on $ (0,y_0)$ and negative on $ (y_0, R) $, we employ Cauchy-Riemann equations.

	Consider the derivative of $ g(z) $, denoted by $ g^{(1)} (z) $:
	\[ 
	g^{(1)}(z) : =\frac{\D }{\dz}g (z) =1 - \frac{z}{\sqrt{z^2-1}}  + \frac{1}{z^2 \sqrt{1-z^ {-2}}}= 1-  \frac{\sqrt{z^2-1}}{z}.
	\]
	Separate the real and imaginary parts of the derivative, i.e., $ g^{(1)} (x + \imagI y) = u^{(1)}(x,y) + \imagI v^{(1)}(x,y) $, $ x,y \in \mbb R $, then    Cauchy-Riemann equations imply
	\[ 
	\begin{cases}
		u^{(1)}(x_0,y_0) =  \frac{\partial  u}{\partial x}  (x_0,y_0)  =  \frac{\partial  v}{\partial y}  (x_0,y_0)   , \\
		v^{(1)}(x_0,y_0) = -  \frac{\partial  u}{\partial y}  (x_0,y_0)  =  \frac{\partial  v}{\partial x}  (x_0,y_0)   ,
	\end{cases}
	\]
	for any  $ (x_0, y_0) \in (1,R) \times (0,R) $. Moreover, these equations remain also true for $ x_0= 1 $, $ y_0 > 0 $ (this follows from the fact that $ g $ is holomorphic in a neighborhood of any $ z=1+ \imagI y_0  $ with $ y_0>0 $). Since we are interested in $ \frac{\partial  v}{\partial y}  (1,y)  $ for $ y \in (0,\infty) $, consider the real part of $ g^{(1)} (z)$, where $ z = 1+\imagI y $. A direct calculation gives
	\[ 
	\frac{\partial  v}{\partial y}  (1,y)
	=
	u^{(1)}(1,y) = 
	1-\frac{\sqrt{\frac{y^2-1}{{\left(y^2+1\right)}^2}+\frac{\sqrt{y^4+4\,y^2}}{y^2+1}+1}}{\sqrt 2}.
	\]
	It is easy to verify that this expression is positive for $ y \in (0,y_0) $ and negative for  $ y\in (y_0,0) $, where $ y_0 >0 $ is the real  solution of 
	\[ 
	\frac{y\,\sqrt{y^2+4}+1}{y^2+1}-\frac{2}{{\left(y^2+1\right)}^2} = 1.
	\]
	Letting $ t = y^2 $ and simplifying   gives the equation $ t^3 + t^2 = 1 $, whose only real solution is 
	\[ 
	t_0 = 
	\frac{{\left(\frac{25}{2}-\frac{3\,\sqrt{69}}{2}\right)}^{1/3} +{\left(\frac{3\,\sqrt{69}}{2}+\frac{25}{2}\right)}^{1/3} - 1}{3} 
	\approx 0.7548\ldots,
	\]
	thus $ y_0 = \sqrt{t_0 } \approx 0.86883\ldots $ and $ \Im (g(z)) $ attains its maximum value at $ z = 1 +\imagI y_0 $. Finally, numerical calculations yield
	$ g(1+\imagI y_0) = 
	1.29797\ldots +\imagI   0.26066\ldots   $,
	thus $ \Im (g(z))  < 0.2607$ for the values of $ z $ under consideration. 
\end{proof}

\begin{remark}
 See the remark at the end of \cref{sec:4.3} for an interpretation of \cref{thm:L06} in the context of the modified Bessel function of the second kind.
\end{remark}

\subsection{Properties of the main function}
It can be noticed that the integrand $  x^{-1}  J_\nu (x)   \cos^n\lr{ \frac{x}{n}  } $ is the real part of
the function $f $ defined by \eqref{eq:f_def}.
Since $ \hankel $ is holomorphic  throughout the complex plane cut along the negative real axis,   $ f $ is holomorphic in the domain $  \lrb{z \in \mbb C \ \vline\  \Re z \geq \epsilon} $, for an arbitrary $ \epsilon > 0 $.

We intend to estimate the integral of $f$ with the help of \cref{thm:L02} by dividing the integration domain into subintervals with endpoints 0, $na_1$, $na_2$, \ldots; to that end, we will make use of the following properties of $f$.

First we characterize and bound the values of $f$:
\begin{lemma}\label{thm:L08b}
	If $x,y\geq 0$ are nonnegative reals  such that $f$ is defined at $x+\imagI y$, then
	\begin{equation}\label{eq:L08.03}
		f(x+\imagI y)
		=
		\lr{  \cos(x/n) \cosh (y /n) - \imagI \, \sin(x/n) \sinh (y/n)  \vbl  }^n \, \frac{ \hankel(x+\imagI y)}{x  + \imagI y}.
	\end{equation}
	Furthermore, if $ \max\{x,y\} \geq n^2 $, then the following inequality is satisfied
	\begin{equation}\label{eq:L08.02}
		\lrv{f(x + \imagI y) }    \leq  3 \lrv{x + \imagI y}^{-3/2} . 
	\end{equation}
\end{lemma}
\begin{proof}
	The equality in \eqref{eq:L08.03} follows from the definition of $f$ and  the identity
	\[ 
	\cos (x/n+ \imagI y/n)=   \cos(x/n) \cosh (y /n) - \imagI \, \sin(x/n) \sinh (y/n).
	\]  
	To show \eqref{eq:L08.02}, 
	apply \eqref{eq:1.2.04} to bound $ \hankel(x + \imagI y) $ in \eqref{eq:L08.03}  (notice that $ \arg(x + \imagI y) \in [0, \pi/2) $ for the $ x,y $ values under the consideration; moreover, $ \lrv {x + \imagI y} \geq   \max\{x,y\} \geq n^2 $),
	obtaining
	\begin{align*}
		\lrv{f(x + \imagI y) }   
		= & 
		\lrv{ \cos(x/n) \cosh (y /n) - \imagI \, \sin(x/n) \sinh (y/n) \vbl    }^n \cdot
		\lrv{  \frac{  3\e^{-y }      }{\lr{x + \imagI y}^{1.5}} }
		\\ 
		\leq  
		&\frac{  3\e^{-y} \lr{\vbl \lrv{ \cos(x/n) \cosh(y/n)} +  \lrv{\sin (x/y)\sinh(y/n) } }^n}{ \lrv{x + \imagI y}^{1.5}} 
		\\ 
		\leq  
		&   3\e^{-y} \lr{\vbl  \cosh(y/n)  +  \sinh(y/n)  }^n \,  \lrv{x + \imagI y}^{-1.5}
		=
		3 \lrv{x + \imagI y}^{-1.5} . 
	\end{align*} 
\end{proof}

We also characterize the values of $f$  when its argument is of the form $na_k + \imagI y$:
\begin{lemma}\label{thm:L09b}
	For all $ y\geq 0 $ and  integer $ k \geq 1 $ it holds that
	\begin{equation}\label{eq:L08.01}
		f(na_k+\imagI y)
		=
		2^{-n} \,  \frac{ \e^{n\imagI    a_{k+1}} \cdot   \lr{\e^{y/n}- \e^{-y/n}  }^n   \hankel( na_k + \imagI y) }{ na_k + \imagI y},
	\end{equation}
	and, whenever $\lrv{na_k + \imagI y} \geq n^2$, 
	\begin{equation}\label{eq:L09.01b}
		\lrv{f(na_k+\imagI y)}
		< 
		3\cdot 2^{-n} \,  \lrv{na_k + \imagI y}^{-3/2}.
	\end{equation}

	Furthermore,
	\begin{enumerate}
		\item if $1 \leq k < n$, then for all   $ y  \geq  0$ the   following estimate holds:
		\begin{equation}
			\lrv{ f(na_k+\imagI y)} < 860 \cdot 1.541^{-n}   \,  \lrv{ n a_k + \imagI y}^{-1.5}.\label{eq:L09.02a}
		\end{equation} 
		\item if $ 2 \leq n \leq  k $, then for all   $ y  \geq  0$ the   following estimate holds:
		\begin{equation}
			\lrv{ f(na_k+\imagI y) -   f(na_{k+1}+\imagI y)}
			<
			15 n^2  \cdot    2^{-n}    \lrv{ na_{k}   + \imagI y }^{-2.5}
			\label{eq:L09.02}
		\end{equation} 
	\end{enumerate}
	
\end{lemma}
\begin{proof}
	Set $ x = na_k $ in  \eqref{eq:L08.03}; since
	$ \cos (a_k)=0 $,  $ \lr{- \imagI \, \sin(a_k)}^n 	= 	\e^{-\imagI   n a_{k} } 	= 	\e^{\imagI   n a_{k+1} }  $ and
	\[ 
	\sinh(y/n)  = 0.5   \lr{\e^{y/n}- \e^{-y/n}  },
	\]
	we obtain \eqref{eq:L08.01}. Now, if $\lrv{na_k + \imagI y} \geq n^2 $, we can apply \eqref{eq:1.2.04} to bound $ \hankel(x + \imagI y) $ in \eqref{eq:L08.01}, 
	obtaining  the desired estimate 
	\eqref{eq:L09.01b}:
	\[
	\lrv{f(na_k+\imagI y)}
	\leq 
	2^{-n} \,    \lr{\e^{y/n}- \e^{-y/n}  }^n     
	\frac{  3\e^{-y }      }{\lrv{na_k + \imagI y}^{1.5}} 
	<
	3 \cdot 2^{-n} \,  \lrv{na_k + \imagI y}^{-1.5}.
	\]

	\paragraph{Proof of \eqref{eq:L09.02a}.}
	Suppose $1 \leq k \leq n-1$.  	Let  $ w = na_k/ \nu + \imagI y/\nu $ and apply \cref{thm:L01} to  rewrite \eqref{eq:L08.01}  as
	\begin{align*}
		\lrv{ f(na_k+\imagI y)}
		=& 
		\frac{     \lrv{\e^{y/n}- \e^{-y/n}  }^n  \lrv{ \hankel( \nu  w) }}{2^{n} \, \lrv{ na_k + \imagI y}}  \\
		\leq  &  
		\frac{     \e^y    \lrv{ \e^{ \imagI \nu \lr{\sqrt{w^2- 1}  - \arccos (1/w)} } }}{2^{n} \,  \sqrt{ \nu  \pi/2 } \lrv{ na_k + \imagI y}  \lrv{1-w^2}^{0.25}}  \cdot \lrv{   1+ \eta(\nu,-\imagI w)}. \\
		=  &  
		\frac{     \e^{y  - \nu \cdot  \Im  \lr{\sqrt{w^2- 1}  - \arccos (1/w)} }  }{2^{n} \,  \sqrt{ \nu  \pi/2 } \lrv{ na_k + \imagI y}  \lrv{1-w^2}^{0.25}}  \cdot \lrv{   1+ \eta(\nu,-\imagI w)}.
	\end{align*}
	Notice that $ y  - \nu \cdot  \Im   \lr{\sqrt{w^2- 1}  - \arccos (1/w)}    = \nu\, \Im (g(w)) $,  where $ g $ is defined as in \cref{thm:L06}. Moreover,  $ \Re w  \geq   na_1 / \nu  >  \pi/2$, so \eqref{eq:L01e02} from \cref{thm:L01}  applies;
	hence we have
	\begin{equation}\label{eq:L10.01}
		\lrv{ f(na_k+\imagI y)} 
		\leq 
		\frac{430 \e^{ \nu\, \Im (g(w))}    }{2^{n} \,  \sqrt{ \nu  \pi/2 } \lrv{ na_k + \imagI y}  \lrv{1-w^2}^{0.25}}  .
	\end{equation}

	Since  $ \lrv{w} \geq  na_1 / \nu   > \pi/2 > 1.5$,   we have $ 4/9\lrv w^2  > 1 $ and by the triangle inequality,
	\[ 
	\lrv{w^2-1} \geq \lrv{w}^2-1 \geq   5/9 \lrv w^2.
	\]
	Thus
	\[ 
	\lrv{w^2-1} ^{0.25} \geq    \lr{  5/9 \lrv{w}^2       }^{0.25}   >    \lrv{0.5 w}^{0.5}= \lrv{  na_k + \imagI y }^{0.5}  \cdot   {(2\nu)}^{-0.5} .
	\]
	This  allows to estimate 
	\[
	\sqrt{ \nu  \pi/2 } \lrv{ na_k + \imagI y}  \lrv{w^2-1}^{0.25}
	> 0.5\lrv{  na_k + \imagI y }^{1.5}   
	\]
	and simplify \eqref{eq:L10.01} to  
	\begin{equation}\label{eq:L10.02}
		\lrv{ f(na_k+\imagI y)}
		<
		\frac{860 \exp\lr{\nu\, \Im (g(w))}}{ 2^{n}   \lrv{  na_k + \imagI y }^{1.5}}.
	\end{equation}
	From \cref{thm:L06} we have   $ \Im (g(w))  < 0.2607$, since  $ \Re w  \geq  \pi/2  > 1 $,   $ \Im w \geq 0  $.
	From  ${\nu<n}$ and $ \e^{0.2607}/2 <  1.541^{-1}$ we obtain
	\[ 
	\frac{\e^{  \nu \Im (g( w)) }}{2^n}
	< 
	\lr{ \e^{0.2607}/2} ^n < 1.541^{-n}
	\]
	which together with \eqref{eq:L10.02} yields the desired bound:
	\[ 
	\lrv{ f(na_k+\imagI y)} < 860 \cdot 1.541^{-n}   \,  \lrv{ n a_k + \imagI y}^{-1.5}.
	\] 
	
	\paragraph{Proof of \eqref{eq:L09.02}.}
	Suppose, $k \geq n \geq 2$.
	We utilise \eqref{eq:1.2.01}  to rewrite \eqref{eq:L08.01} as
	\begin{equation}
		f(na_k+\imagI y)
		=
		\frac{    2^{-n } \lr{1 - \e^{-2y/n}  }^n  \e^{ - i\pi ( 2\nu +1)/4}   }{\sqrt{ \pi /2}}  \cdot \frac{    1 + \rho(\nu, na_k + \imagI y)   }  {( na_k + \imagI y)^{3/2}},   \label{eq:L09.01}
	\end{equation} 
	since  $\imagI (n a_k + \imagI y) = \imagI n a_k - y $ and
	\[ 
	\e^{\imagI   n a_{k+1}} \cdot   \e^{ \imagI (na_k  - \nu \pi/2 - \pi/4)} 
	= \e^{ - i\pi ( 2\nu +1)/4}
	\ \text{and}\  
	\e^{-y}  \lr{\e^{y/n}- \e^{-y/n}  }^n=    \lr{1- \e^{-2y/n}  }^n .
	\] 
	Now,  using \eqref{eq:L09.01}  we can bound
	\begin{equation}
		\lrv{ f(na_k+\imagI y) -   f(na_{k+1}+\imagI y)}
		\leq 
		2^{-n} \lrv{ 
			\frac{   1 + \rho(\nu, na_k + \imagI y)   }  {( na_k + \imagI y)^{3/2}}
			-
			\frac{   1 + \rho(\nu, na_{k+1} + \imagI y)   }  {( na_{k+1} + \imagI y)^{3/2}}
		}.
		\label{eq:L09.03}
	\end{equation} 
	Since $  na_k \geq n^2$ and $\nu<n$,  from  \eqref{eq:1.2.03}   we have
	\[ 
	\lrv{ \rho(\nu, na_{k} + \imagI y) } \cdot   \lrv{ na_{k}   + \imagI y }
	<n^2 \e ;
	\]
	similarly,
	\[ 
	\lrv{ \rho(\nu, na_{k+1} + \imagI y) } \cdot \lrv{ na_{k+1}   + \imagI y }  < n^2 \e .
	\]
	Therefore we can estimate the RHS of \eqref{eq:L09.03} via 
	\begin{align*} 
		&\lrv{ 
			\frac{   1 + \rho(\nu, na_k + \imagI y)   }  {( na_k + \imagI y)^{3/2}}
			-
			\frac{   1 + \rho(\nu, na_{k+1} + \imagI y)   }  {( na_{k+1} + \imagI y)^{3/2}}
		}
		\\
		&  
		\leq 
		\lrv{ 
			\frac{   1   }  {( na_k + \imagI y)^{3/2}}
			-
			\frac{   1  }  {( na_{k+1} + \imagI y)^{3/2}}
		}
		+ 
		\lrv{ 
			\frac{   \rho(\nu, na_k + \imagI y)   }  {( na_k + \imagI y)^{3/2}}
		}
		+
		\lrv{ 
			\frac{     \rho(\nu, na_{k+1} + \imagI y)   }  {( na_{k+1} + \imagI y)^{3/2}}
		}
		\\
		& < \lrv{ 
			\frac{   1   }  {( na_k + \imagI y)^{3/2}}
			-
			\frac{   1  }  {( na_{k+1} + \imagI y)^{3/2}}
		}
		+ \frac{n^2 \e }{\lrv{na_k + \imagI y}^{5/2} }+ \frac{n^2 \e }{\lrv{na_{k+1} + \imagI y}^{5/2} }
	\end{align*}  
	or (since $na_{k+1} > na_k$)
	\begin{equation}
		2^n\lrv{ f(na_k+\imagI y) -   f(na_{k+1}+\imagI y)}
		<
		\lrv{ 
			( na_k + \imagI y)^{-3/2}
			-
			( na_{k+1} + \imagI y)^{-3/2}
		}
		+\frac{2n^2 \e }{\lrv{na_k + \imagI y}^{5/2} } .
		\label{eq:L09.04}
	\end{equation} 
	By the mean value theorem (applied to the function $x \mapsto  (x+\imagI y)^{-3/2} $ ),
	\[ 
	\frac{   1   }  {( na_k + \imagI y)^{3/2}}
	-
	\frac{   1  }  {( na_{k+1} + \imagI y)^{3/2}} 
	=
	\frac{3\pi n}{2 ( \xi  + \imagI y)^{5/2}}, 
	\]
	for some $\xi \in (na_k, na_{k+1})$, which gives us 
	\[
	\lrv{ 
		\frac{   1   }  {( na_k + \imagI y)^{3/2}}
		-
		\frac{   1  }  {( na_{k+1} + \imagI y)^{3/2}}
	}
	< \frac{1.5  \pi n}{\lrv{ na_k + \imagI y}^{5/2}}.
	\]
	Combining this with   \eqref{eq:L09.04} yields \eqref{eq:L09.02}:
	\[
	\lrv{ f(na_k+\imagI y) -   f(na_{k+1}+\imagI y)}
	\leq 
	2^{-n} 
	\cdot 	\frac{   1.5 \pi n + 2 \e n^2  }  {\lrv{ na_k + \imagI y}^{5/2}}
	< \frac{15n^2} {2^n\lrv{ na_k + \imagI y}^{5/2}}.
	\] 
\end{proof}

\subsection{Proof of   Theorem 2} \label{sec:A.5} 
Now we can prove \cref{thm:L11}, restated here for convenience:
\thesecondthm*

\begin{proof}[Proof of  \cref{thm:L11}]
	Denote
	\begin{align*}
		& I_0(n,\nu) = \int_0^{n a_1}  x^{-1}  J_\nu (x)   \cos^n  \lr{\frac{x}{n} } \dx ;\\
		& I_k(n,\nu) = \int_{n a_k}^{n a_{k+1}}  x^{-1}  J_\nu (x)   \cos^n \lr{\frac{x}{n}} \dx ,  \quad  k \in \mbb N.
	\end{align*}
	
	\paragraph{The tail integral (proof of \eqref{eq:L11.01}).}
	Let $ k \geq n \geq 2 $ and notice that   $ I_k(n,\nu) $ is the real part of the integral 	$  \int_{n a_k}^{n a_{k+1}} f(x)  \dx   $.
	The function   $ f $ is  holomorphic in a domain containing  $ D_{n a_k,  na_{k+1}} $; let us verify the other assumptions of   \cref{thm:L02}.
	
	Since  $  na_k \geq  n (n-0.5) \pi > n^2 $,  \cref{eq:L08.02}  in
	\cref{thm:L08b} implies
	\[ 
	\sup_{x \in [n a_k,  na_{k+1}]}  \lrv{ f(x + \imagI R)}     \leq   \frac{3}{ \lrv{ na_k + \imagI R}^{1.5}}  
	\]	
	and therefore the assumption $ \lim\limits_{R \to +\infty}  \sup_{x \in [a,b]}  \lrv{ f(x + \imagI R)}   = 0 $ of \cref{thm:L02} is satisfied.
	
	From \eqref{eq:L09.01b} it follows that  $\lrv{f(na_k + \imagI y )} \leq 3 \cdot 2^{-n}\lrv{na_k + \imagI y}^{-3/2} $, thus  the integral $ \int_0^ \infty  f(na_k+\imagI y) \dy  $ converges and its absolute value, by \cref{thm:L03}, \cref{eq:L03.01}, is  bounded by 
	\[
	 \int_0^ \infty  \lrv{f(na_k+\imagI y) } \dy  
	 \leq 
	{3 \cdot 2^{-n} \int_{0}^{\infty} \frac{\dy }{  (n^2a_k^2 + y^2)^{3/4} }
		< \frac{9 \cdot 2^{-n} }{\sqrt{n a_k}}};
	\]
	 similarly, $ \int_0^ \infty  f(na_{k+1}+\imagI y)   \dy $ converges.

	 Now we see that all  assumptions of \cref{thm:L02} are satisfied, therefore we can apply \eqref{eq:L02e01} to obtain
	 \begin{equation}\label{eq:A.16}
	     \lrv{I_k(n,\nu) } \leq 
	     \lrv{\int_{n a_k} ^{n a_{k+1}}  f(x) \dx } 
	     \leq   \int_0^ \infty  \lrv{f(na_k+\imagI y)  -  f(n a_{k+1}+\imagI y) }\dy .
	 \end{equation}
	 Now apply  \cref{thm:L09b}, \cref{eq:L09.02} to upper-bound the difference of $f$ values as
	 \begin{equation}\label{eq:A.17}
	     \lrv{f(na_k+\imagI y)  -  f(n a_{k+1}+\imagI y) }
	 \leq 
	 \frac{15  n^2   }{2^n   } \,   \lrv{ na_{k}   + \imagI y }^{-5/2} .
	 \end{equation}
	 By \cref{thm:L03}, \cref{eq:L03.02} we have
	 \begin{equation}\label{eq:A.18}
	       \int_0^ \infty  \frac{\dy }{(n^2 a_{k}^2   + y^2 )^{5/4}} < \frac{2}{\sqrt{n^3 a_k^3}} 
	       <  \frac{30   \sqrt{n} }{2^n   k^{3/2} },
	 \end{equation}
	 where the last inequality uses $k \leq a_k $.
	 
	 Now, by combining \eqref{eq:A.16}-\eqref{eq:A.18}, we arrive at
	\[ 
	\lrv{ \int_{n a_n}^\infty x^{-1}  J_\nu (x)   \cos^n\lr{ \frac{x}{n}  } \dx   }   \leq \sum_{k=n}^  \infty  \lrv{I_k(n,\nu) } 
	<   \frac{30       \sqrt{n} }{2^n    } \sum_{k=n}^  \infty    k^{-3/2} .
	\]
	Since $30 \sum_{k=n}^  \infty    k^{-3/2}  < \sum_{k=1}^  \infty    k^{-3/2}  \approx 30\cdot2.6124 \ldots < 100$, 
	we are done.

	\paragraph{The middle part (proof of \eqref{eq:L11.02}).}
	Let $ k   $ satisfy $ 1 \leq k \leq n-1 $ and notice that   $ I_k(n,\nu) $ is the real part of the integral 	$  \int_{n a_k}^{n a_{k+1}} f(x)  \dx   $.
	The function   $ f $ is  holomorphic in a domain containing  $ D_{n a_k,  na_{k+1}} $; let us verify the other assumptions of   \cref{thm:L02}.
	
	When $ R>n^2 $,  \cref{thm:L08b} implies
	\[ 
	\sup_{x \in [n a_k,  na_{k+1}]}  \lrv{ f(x + \imagI R)}     \leq   \frac{3}{ \lrv{ na_k + \imagI R}^{1.5}}  
	\]	 
	and therefore the assumption $ \lim\limits_{R \to +\infty}  \sup_{x \in [a,b]}  \lrv{ f(x + \imagI R)}   = 0 $ of \cref{thm:L02} is satisfied.
	
	Split 
	\[ 
	\int_0^ \infty  f(na_k+\imagI y)  \dy =  \int_0^ {n^2}    f(na_k+\imagI y)\dy   +  \int_{n^2}^ \infty  f(na_k+\imagI y) \dy ,
	\]
	then   from \eqref{eq:L09.01b} in \cref{thm:L09b} we conclude that 
	$ \int_{n^2}^ \infty  f(na_k+\imagI y) \dy  $
	converges and its absolute value is bounded by
	\[ 
	3 \cdot 2^{-n}\int_{0}^{\infty} \frac{\dy }{  (n^2a_k^2 + y^2)^{3/4} }
	< 
	9 \cdot 2^{-n}   \lr{n a_k}^{-0.5}.
	\]
	Thus $  \int_0^ \infty  f(na_k+\imagI y)  \dy $ converges as well  (and so does  $  \int_0^ \infty  f(na_{k+1}+\imagI y)  \dy $) and  \cref{thm:L02} applies. From \eqref{eq:L02e01}
 we now obtain 
	\[ 
	\lrv{I_k(n,\nu)} 
	\leq 
	\lrv{\int_0^ \infty  f(na_k+\imagI y)  \dy }
	+
	\lrv{\int_0^ \infty  f(na_{k+1}+\imagI y)  \dy }
	.
	\]

	Moreover, we also conclude 
	\begin{equation}\label{eq:L11.04}
		\lrv{\int_0^ \infty  f(na_k+\imagI y)  \dy } <  \lrv{ \int_0^ {n^2}    f(na_k+\imagI y)\dy  } +\frac{ 8}{2^n \sqrt{n}  }.
	\end{equation}
	Here we have used  $ a_k \geq a_1 = \pi/2 $ and
	\[  
	\lrv{\int_{n^2}^ \infty  f(na_k+\imagI y) \dy  }
	<
	\frac{9 }{2^n \sqrt{n a_1}}
	=
	\frac{9 \sqrt 2 / \sqrt \pi }{   2^n \sqrt{n}    } 
	<
	\frac{8}{    2^n \sqrt{n} } .
	\]

	Due to  \cref{thm:L09b}, \cref{eq:L09.02a}, we can estimate the remaining integral as
	\[ 
	\lrv{ \int_0^ {n^2}    f(na_k+\imagI y)\dy  }
	<
	860 \cdot 1.541^{-n}  \,  \int_0^ {\infty} \lrv{ n a_k + \imagI y}^{-1.5} \dy 
	< 
	\frac{1800  }{ 1.541^{n} \, \sqrt{n  }} .
	\]
	The last step applies \eqref{eq:L03.01} and inequalities 
	\[ 
	860 \cdot \frac{B(0.5, 0.25) } {  2\sqrt{a_k}}  
	\leq 
	860 \cdot \frac{B(0.5, 0.25) } {  \sqrt{2\pi}}
	< 
	1800. 
	\]
	Now \eqref{eq:L11.04} gives
	\[ 
	\lrv{\int_0^ \infty  f(na_k+\imagI y)  \dy } 
	<
	\frac{1800}{ 1.541^{n} \, \sqrt{n   }}
	+
	\frac{ 8}{2^n \sqrt{n}  }
	< \frac{2000}{ 1.541^{n} \, \sqrt{n  }},
	\]
	thus
	\[ 
	\lrv{I_k(n,\nu)} 
	\leq 
	\lrv{\int_0^ \infty  f(na_k+\imagI y)  \dy }
	+
	\lrv{\int_0^ \infty  f(na_{k+1}+\imagI y)  \dy }
	<
	\frac{4000}{ 1.541^{n} \, \sqrt{n  }}.
	\]
	
	Finally,    \eqref{eq:L11.02} is obtained as
	\[ 
	\lrv{ \int_{n a_1}^{na_n} x^{-1}  J_\nu (x)   \cos^n\lr{ \frac{x}{n}  } \dx   }
	\leq \sum_{k=1}^{n-1} 	\lrv{I_k(n,\nu)} 
	< \frac{4000 \sqrt n}{1.541^{n} }.
	\]

	\paragraph{The bulk (proof of \eqref{eq:L11.03}).}
	Denote $ c = \nu / n  \in (\ratioLower,1)$.
	We start by splitting
	\begin{equation}\label{eq:L11.05}
		I_0(n,\nu)  =
		\int_0^{\nu}  x^{-1}  J_\nu (x)   \cos^n  \lr{\frac{x}{n} } \dx 
		+
		\int_{nc}^{n \pi /2}  x^{-1}  J_\nu (x)   \cos^n  \lr{\frac{x}{n} } \dx .
	\end{equation}
	The second integral can be bounded by 
	employing  
	the fact that $ \cos(x/n)^n / x $ is decreasing in $ [nc ,n\pi/2]  $
	and
	noting that \cite[\href{https://dlmf.nist.gov/10.14.E1}{Eq. 10.14.1}]{DLMF}  the absolute value of the Bessel function $ J_\nu $ is bounded  by 1 for all real arguments, whence
	\begin{equation}\label{eq:L11.06}
		\int_{nc}^{n \pi /2}  x^{-1}  J_\nu (x)   \cos^n  \lr{\frac{x}{n} } \dx 
		\leq  
		\frac{\cos^n  (c)}{nc} \cdot  (n\pi/2-nc)
		<
		2\cos^n(\ratioLower).
	\end{equation}
	The last step relies on the assumption $ \ratioLower > \pi/6 $.

	For the first integral in the RHS of \eqref{eq:L11.05}, we apply \cref{thm:L04}, which gives $\cos^n  \lr{\frac{x}{n} } \leq \exp(-\frac{x^2}{2n})$, as well as the estimate
	\[
	\lrv{J_\nu(\nu t) }  \leq  J_\nu(\nu)\cdot  t^\nu  \exp \lr{\frac{\nu^2 (1-t^2)}{2\nu+4}},
	\] 
	valid \cite[p.204]{RBParis}  for all $ t\in (0,1] $ and $\nu>0$. This inequality can be slightly weakened by bounding $\frac{\nu^2 (1-t^2)}{2\nu+4} < \frac{\nu (1-t^2)}{2}$.
	Those imply that the integral  can be bounded as
	\begin{align*}
		\lrv {\int_0^{\nu}  x^{-1}  J_\nu (x)   \cos^n  \lr{\frac{x}{n} } \dx }
		\leq 
		&
		J_\nu(\nu) \nu^{-\nu} \cdot  
		{
			\int_0^{\nu } 
			x^{\nu-1}  \exp \lr{ \frac{\nu}{2} - \frac{x^2}{2\nu} - \frac{x^2}{2n}  }   \dx 
		}\\ 
		= &
		J_\nu(\nu) \lr{\frac{ \e}{\nu^2}}^{\nu/2} \cdot  
		{
			\int_0^{\nu } 
			x^{\nu-1}  \exp \lr{   - \frac{(c+1)x^2}{2\nu}   }   \dx 
		}\\ 
		= &
		J_\nu(\nu)\lr{\frac{ \e}{\nu^2}}^{\nu/2} \cdot  
		{
			\int_0^{(c+1)\nu/2 } 
			\lr{\frac{2\nu y}{c+1}}^{\nu/2-1} \, \frac{\nu}{(c+1)}   \e^{-y}
			\dy 
		}\\ 
		= &
		\frac{1}{2}
		J_\nu(\nu)\lr{\frac{2 \e}{\nu(c+1)}}^{\nu/2} \cdot  
		{
			\int_0^{(c+1)\nu/2 }  
			y^{\nu/2-1}\e^{-y}  \dy 
		}\\ 
		<  &
		0.5 J_\nu(\nu)\lr{\frac{2 \e}{\nu(c+1)}}^{\nu/2} \cdot  
		\underbrace{
			\int_0^{\infty}  
			y^{\nu/2-1} \e^{-y}  \dy 
		}_{\Gamma(\nu/2)}.
	\end{align*}
	Now we apply a bound on the  gamma function, valid \cite[\href{https://dlmf.nist.gov/5.6.E1}{Eq. 5.6.1}]{DLMF} for positive arguments:
	\[
	\Gamma(x)>\sqrt{\frac{2\pi}{x}} x^x \exp\lr{-x + \frac{1}{12x}} , \quad x>0.
	\]
	Taking   $x=\nu/2$ this implies that 
	\[
	\Gamma(\nu/2)>2\sqrt{\frac{\pi}{\nu}} \lr{\frac \nu 2}^{\nu/2} \exp\lr{-\frac{\nu}2 + \frac{1}{6\nu}} 
	\]
	and, since $\nu>1$,
	\[
	\lrv {\int_0^{\nu}  x^{-1}  J_\nu (x)   \cos^n  \lr{\frac{x}{n} } \dx }
	<
	J_\nu(\nu)
	\sqrt{\frac{\pi}{\nu}}  \e^{ 1/{6}}
	\lr{{c+1}}^{-\nu/2} .
	\]
	By 
	\cite[\href{https://dlmf.nist.gov/10.14.E2}{Eq. 10.14.2}]{DLMF},
	\[ 
	0 < J_\nu(\nu) < \frac{2^{1/3}}{3^{2/3} \Gamma(2/3) \nu ^{1/3} } < \frac{0.45}{\nu^{1/3}}  .
	\]
	From $0.45\sqrt\pi  \e^{ 1/{6}} < 1$ and $\nu  > 1$ we arrive at
	\begin{equation}\label{eq:L11.07}
		\lrv {\int_0^{\nu}  x^{-1}  J_\nu (x)   \cos^n  \lr{\frac{x}{n} } \dx }
		<
		\frac{0.45 \sqrt \pi  \e^{ 1/{6}} }{\nu^{5/6}} 
		\exp\lr{ - \frac{n c\ln(c+1)}{2}   }
		<
		(1+c)^{-cn/2}.
	\end{equation}
	We note that $ c \mapsto (1+c)^{c/2}$ is increasing in $c$, therefore   
	the bound on the RHS  of \eqref{eq:L11.07} is decreasing in $c$ and attains its maximal value at $\ratioLower$. 
	
	Finally, an application of \cref{thm:L04} gives $ \cos^n(\ratioLower) < \exp(- 0.5 n \ratioLower^2) $;  since  $ -\ratioLower^2 <  - \ratioLower\ln(\ratioLower+1)$, we have 
	\[
	2\cos^n(\ratioLower) < 2(1+\ratioLower)^{-\ratioLower n/2},
	\]
	which yields the desired inequality  \eqref{eq:L11.03}. 
\end{proof}

\subsection{The Hankel function's expansion}\label{sec:4.0}

\begin{proof}[Proof of \cref{thm:L01}]
We proceed to prove the asymptotic expansion of $\hankel(z)$ when the argument and order are of similar magnitude.
We   relate the Hankel function to $ K_\nu(z) $, the modified Bessel function of the second kind, via    \cite[\href{https://dlmf.nist.gov/10.27.E8}{Eq. 10.27.8}]{DLMF}  
\begin{equation}\label{eq:hankel_K}
	K_{\nu}(z)= 0.5 \pi  \e^{(\nu+1)\pi \imagI/2} \hankel (\imagI z) ,
	\quad  
	- \pi \leq \arg z \leq  \pi /2 .
\end{equation}
In \cref{sec:4.1} we will investigate the expansion of $ K_{\nu} (\nu   z) $ due to Olver \cite[Chapter 10]{Olver97b} and sketch a brief overview of the techniques employed. 
These bounds  \eqref{eq:L01.03}-\eqref{eq:L01.04} are not yet explicit in the sense that  the bound \eqref{eq:L01.04}  depends on an unspecified  variational path and has to be estimated.
In  \cref{sec:4.2}  we make the bounds explicit. Finally, in \cref{sec:4.3} we use \eqref{eq:hankel_K} to derive the explicit bounds on the Hankel function.

\subsubsection{Error bounds for the modified Bessel function}\label{sec:4.1}
Denote
\[
 \xi(z) =  (1+z^2)^{1/2} + \ln \frac{z}{1+ \sqrt{1+z^2}},
\] 
where all branches take take their principal values on the positive real axis and are continuous elsewhere.
Consider $ K_{\nu} (\nu   z) $, when  $ \nu >0  $ and $ \lrv{\arg z} < \pi /2 $. Then
	\cite[Chapter 10, Eq. (7.17) \& Eq. (7.15)]{Olver97b}
	\begin{align}
		& K_\nu (\nu z)  = 
		\lr{ \frac{\pi}{ 2\nu } }^{0.5} \frac{\e^{- \nu \xi(z) }}{(1+z^2)^{0.25}} \lr{ 1+ \eta(\nu,z)} \label{eq:L01.03}\\
		& \lrv{ \eta(\nu,z)}   \leq \exp \lr{ \frac{2 \mc V_{+\infty,z}}{ \nu}  } \, \frac{2 \mc V_{+\infty,z}}{ \nu}  , \label{eq:L01.04}
	\end{align}
with  $ \mc V_{+\infty,z} \in \mbb R $  explained later.
	These  estimates  are derived by approximating the solutions of the differential equation \cite[Eq. (7.02), p. 374]{Olver97b},
	\[
	\frac{\D^2 w}{\D z^2}  = \lr{\nu^2  \frac{1+z^2}{z^2}  - \frac{1}{4z^2} }w,
	\]
	which is satisfied by  $ z^{1/2} K_\nu (\nu z)   $; in the differential equation $z$ is confined to the half-plane $\Re z  > 0$.  In the subsequent analysis, change of variables $z \mapsto \xi(z)$ takes place,  mapping the half-plane $ \Re z > 0 $ to a region in the  $\xi $ plane consisting of the half-plane $\Re \xi > 0$ and  the half-strip $ \lrv{\Im \xi } < \frac{1}{2}\pi, \Re \xi \leq 0 $.
	
	\begin{figure}[ht]
		\begin{center}
			\includegraphics[width=\textwidth]{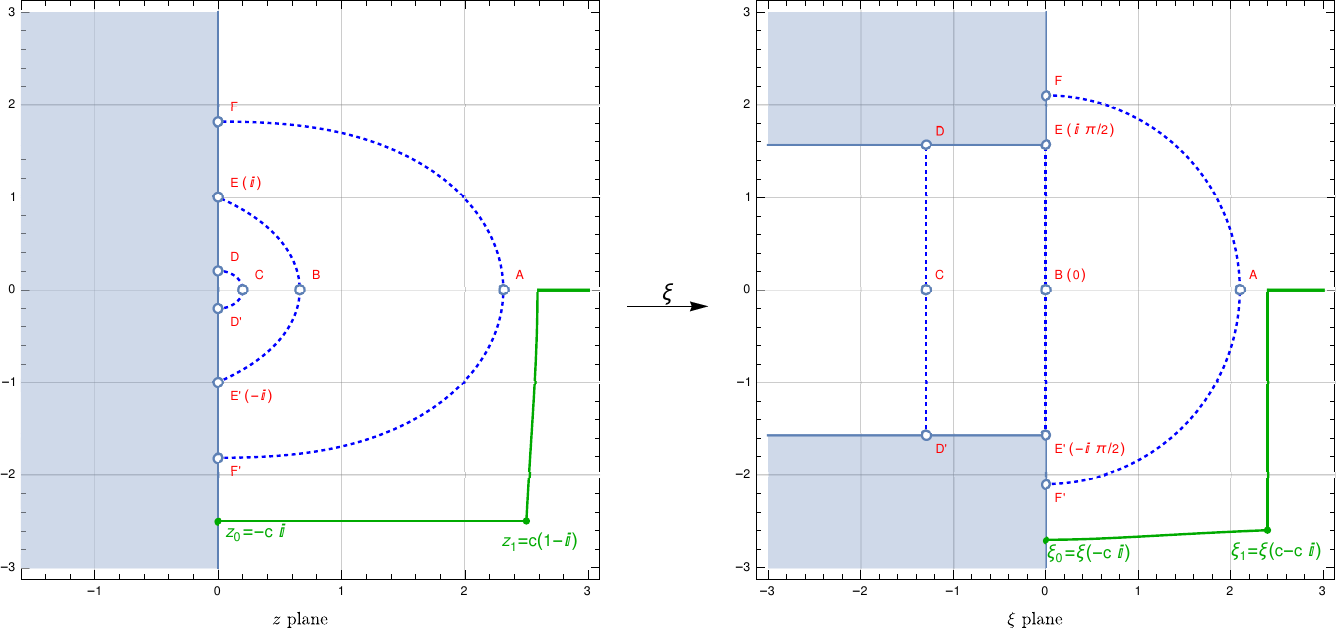}
			\caption{Sketch of the variational path}\label{fig:varpath}
		\end{center}
	\end{figure}
	In \cref{fig:varpath} we illustrate both regions; note that \cref{fig:varpath} essentially  reproduces Figs. 7.1-7.2 from \cite[p.376]{Olver97b} (except for the variational path). On the left, the $z$ plane and the half-plane $\Re z>0$ is shown; on the right, the $\xi$ plane with the image of  $\Re z > 0$ is illustrated.  In the $\xi$ plane, we depict a few contours and points on them; on the left, the preimages of these contours and points are shown. The shaded regions are the ``shadow regions'' in Olver's terminology.
	
	The quantity $ \mc V_{+\infty,z} $ appearing in  \eqref{eq:L01.04}  is defined as  the total variation of \cite[Eq. (7.11), p.376]{Olver97b} the function 
	\begin{equation}\label{eq:u1def}
		U_1(p) = \frac{3p - 5p^3}{24}, \quad p:=(1+z^2)^{-1/2} ,
	\end{equation}
along any \emph{$ \xi $-progressive} path connecting $ +\infty $ with $ z : \Re z>0 $.
	A path $ \gamma $ is said to be $ \xi $-progressive \cite[p.222]{Olver97b} if 1)  $ \gamma $ is a piecewise $ C^2 $-path and 2) $ \Re  \xi(\gamma)$ is non-increasing 
	as $ \gamma $ passes from  $ +\infty  $ to $ z $.

	The case $ \Re z > 0 $ is not sufficient for us, however, since \eqref{eq:hankel_K} effectively rotates the argument of $ \hankel  $ in the complex plane by $ - \pi /2$. In order to prove \eqref{eq:L01e01}  also for  real $ w $, we need the expansion of $ K_{\nu}(\nu z) $ also when $ \arg z = - \pi /2  $.
	Fortunately,     the estimates \eqref{eq:L01.03}-\eqref{eq:L01.04} remain valid \cite[Chapter 10, \S 8.2]{Olver97b} also for $ \arg z = -\pi/2 $, provided that $ \lrv z  $  is bounded away from 1 and the variational path for  $   \mc V_{+\infty,z} $ is correctly constructed (i.e., the path is $ \xi $-progressive).

	Typically the variational path for  $   \mc V_{+\infty,z} $ is chosen so that in the $ \xi $  plane  the image of the path travels\footnote{In fact, the construction describes the \emph{reverse} path connecting $ z $ and $ +\infty $; along the described path $ \Re \xi(z) $ must be non-decreasing. This distinction is unimportant for the value of the variation and from now on we shall ignore it.}  from $ \xi(z) $ parallel to the imaginary axis until the real axis is reached, then proceeding along the real axis to $ +\infty$, see, e.g., \cite[Chapter 10, \S 7.5]{Olver97b},  \cite[p. 764]{Setti} or  \cite[p. 2133]{Bao}.
	
	However, when $z$ might be of the form $ z=-c\imagI $ for some $ c>1 $ (as in our setting), this approach is not suitable, since the path must avoid the point $ z=-\imagI $. Instead, for $ z=-c\imagI $ we form a path as follows: travel parallel to the real axis until $ z_1 = c(1-\imagI) $ is reached (satisfying $ \arg z_1 = -\pi/4 $), then proceed from $ z_1 $ as described previously. The path is sketched in \cref{fig:varpath} (on the left), with its image on the right.

	In \cref{sec:4.2} we show that the described path  is indeed  $ \xi $-progressive and estimate the total variation along this path.

	\subsubsection{Explicit error bounds for the modified Bessel function}\label{sec:4.2}
	
	Fix any $ c>1 $ and define $ \gamma_0(s) = c(s - \imagI) $, $s \in [0,1] $ and $ z_0=\gamma_0(0) = -c \imagI$, $ z_1= \gamma_0(1)   =c (1-\imagI)$.
	We shall show that
	\begin{enumerate}
		\item $ \Re \xi(\gamma_0(s)) $ is  non-decreasing (i.e., the described path from  $ +\infty $ to $ z $ is valid);
		\item $ \mc V_{z_1,z_0} \leq  \frac{c^2(c^2+2)}{\sqrt 8 (c^2-1)^{2.5}}$, where $  \mc V_{z_1,z_0} $ is the variation of $ U_1 $ along $ \gamma_0 $.
	\end{enumerate}
Since the described path connects $ z_0 $ to $ z_1 $ and $ z_1 $ to $ +\infty $ with $ \Re \xi (z) $ is non-decreasing (and the path is clearly piecewise $ C^2 $), we conclude that the construction ensures a $ \xi $-progressive path. Furthermore, the variation $ \mc V_{+\infty,z_1}      $ can \cite[Eq. (5.13)]{Bao} be bounded as $  \frac{1}{12} + \frac{1}{6 \sqrt 5} + \lr{\frac{4}{27}}^{1/4}  $. Therefore, since $ \mc V_{+\infty,z_0} = \mc V_{+\infty,z_1} + \mc V_{z_1,z_0}    $, we can  estimate $ \mc V_{+\infty,z_0} $ as  
	\begin{equation}\label{eq:L01.05}
		\mc V_{+\infty,z_0}  \leq \frac{1}{12} + \frac{1}{6 \sqrt 5} + \lr{\frac{4}{27}}^{1/4} + \frac{c^2(c^2+2)}{\sqrt 8 (c^2-1)^{2.5}}, 
		\qquad   c:=  -\Im z_0 > 1.
	\end{equation} 
	Finally, the estimate \eqref{eq:L01.05} remains valid when $ z_0 = -c \imagI  $ is replaced by $ z =x -c \imagI  $ with any $ x >0 $:
	\begin{itemize}
		\item If $ x \in (0,c] $, then $ z $ lies on  the described path from $ z_0 $ to $ +\infty $, therefore $ \mc V_{+\infty,z} $ is upper-bounded by $ \mc V_{+\infty,z_0} $;
		\item if $ x > c $, then $ \arg z  \in (-\pi/4, 0)  $ and the bound  \cite[Eq. (5.13)]{Bao} applies; then,  $ \mc V_{+\infty,z} $ is upper-bounded by $  \frac{1}{12} + \frac{1}{6 \sqrt 5} + \lr{\frac{4}{27}}^{1/4}   $.
	\end{itemize}

\paragraph{The path is \texorpdfstring{$ \xi $}{xi}-progressive.}
Since $ \xi $ is symmetric around the real axis, i.e., $ \xi(\bar z) = \overline{\xi(z)}  $, we have $ 2\Re \xi(z) = \xi(z) + \xi(\bar z) $.
Define $ h(s)  =     \xi( c(s - \imagI)  )+\xi( c(s+ \imagI)  ) $, then 
 $ 2\Re \xi(\gamma_0(s)) =h(s)$ and we need to show that $ h : [0,1]\to \mbb R $  is non-decreasing.

 Let  $ s\in [0,1] $; denote $ z =\gamma (s)= c(s-\imagI) $ and  $ \omega = \sqrt{1 + z^2} $.
 Since\footnote{In fact, in \cite{Olver97b} this is the defining property of $ \xi $.} $ \frac{\D \xi}{\D z}  = \frac{\sqrt{1+z^2}}{z}$,  we find that
 \begin{align*}
 	h'(s)  &   = c    \frac{\D \xi}{\dz }\lr{c(s - \imagI) }  + c \frac{\D \xi}{\dz }\lr{c(s + \imagI) } 
 	= c\frac{\omega}{z} + c\frac{\bar \omega }{\bar z}  = \frac{\omega \bar z + z \bar \omega }{c(s^2+1)}   = \frac{2 \Re (\omega \bar z)}{c(s^2+1)} .
  \end{align*}
Since $z = c(s-\imagI)$ satisfies $ \arg(z) \in [-\pi/2, - \pi/4] $, it follows that $\arg (z^2) \in [-\pi, - \pi/2]$ and\footnote{To exclude the possibility $\arg(z^2+1)=0$, one must also take into account that $\lrv z \geq c>1$.} $\arg(z^2+1) \in [-\pi, 0)$. Consequently,  $\arg(\omega)  \in [-\pi/2,0) $.  Since
$ \arg (\bar z) \in [\pi/4,  \pi/2] $, we obtain
$\arg(\omega \bar z)  = \arg(\omega) + \arg(z)  \in [-\pi /4,  \pi /2)$,
thus $ \Re (\omega \bar z) > 0 $. We see that $ h'(s)  $ is positive for $ s\in [0,1] $, and $ h(s) =2\Re \xi(\gamma_0(s)) $ is non-decreasing as required.

\paragraph{Estimate of the variation \texorpdfstring{$ \mc V_{z_1,z_0}$}{V z1,z0}.}

It is worth  recalling that, 
for   a holomorphic function $f$ in a complex domain $ D $ containing a piecewise continuously differentiable path   $ \gamma(s) $, $s \in [s_0, s_1] $,
the total variation of $ f $ along the path $ \gamma  $ is defined as
\[ 
\mc V_{\gamma}(f) =  \int_{s_0}^{s_1} \lrv{ f'(\gamma(s)) \gamma'(s) } \ds .
\]
We have
\[ 
\mc V_{z_1,z_0} = 
\int_0^1 
\lrv{ \frac{\D U_1}{\D z} \lr{p(z)}   }_{	z= \gamma_0(s)} 
\lrv{ \gamma_0'(s) } \ds.
\]
Since
\[
\frac{\D U_1(p) }{\D p}=  \frac{1-5p^2}{8},
\quad
\frac{\D p (z) }{\D z}= \frac{-z}{(z^2+1)^{3/2}}
\quad\text{and}\quad
\gamma_0'(s) = c ,
\]
the integrand equals
\[ 
\frac{1}{8} \, \lrv{ 1 - 5(1+z^2)^{-1} } 
\cdot 
\frac{\lrv z}{\lrv{z^2+1}^{-3/2}}
\cdot c
=
\frac{c}{8} \cdot  \frac{\lrv{z (z^2-4)}}{ \lrv{z^2+1}^{5/2}},
\quad z:=  c(s - \imagI ) .
\] 
From $ \lrv {z \pm \imagI} =\lrv{cs - (c \mp 1) \imagI } \geq  c \mp 1$  
we have $ \lrv{z^2+1}^{5/2} \geq (c^2-1)^{2.5} $. On the other hand,
$ \lrv z \leq c \lrv{1-\imagI}  = c \sqrt 2$ and $ \lrv{z^2-4} \leq \lrv z^2 + 4 \leq 2c^2 +4 $. Since the obtained bound is independent of $ z $,   the integral   satisfies
\[ 
\mc V_{z_1,z_0} \leq   \frac{  \sqrt 2 \, c^2(2c^2+4)}{8 (c^2-1)^{2.5}  }
=
\frac{c^2(c^2+2)}{\sqrt 8 (c^2-1)^{2.5}},
\]
as claimed.

\subsubsection{Explicit error bounds for the Hankel function}\label{sec:4.3}

Let $ w \in \mbb C $ be such that  $ \arg w \in [0; \pi /2) $, then  $ z = -\imagI w $ satisfies $ \arg z \in  [-\pi /2; 0) $ and \eqref{eq:hankel_K}  with \eqref{eq:L01.03} imply
\begin{equation}\label{eq:L01.06}
	\hankel(\nu w) 
	=
	\frac{2}{\pi } \e^{- (\nu + 1) \pi  \imagI /2}K_\nu (\nu z) 
	=
	-\imagI \sqrt{\frac{2}{ \pi \nu }} \frac{\e^{- \nu (\xi(-\imagI w) + \imagI \pi /2 ) }}{(1-w^2)^{0.25}} \lr{ 1+ \eta(\nu,-\imagI w)}.
\end{equation}
To see that  we can rewrite \eqref{eq:L01.06} as \eqref{eq:L01e01},    simplify  $ \xi(-\imagI w)   $ as 
\begin{align*}
	\xi(-\imagI w)  &  
	=  (1-w^2)^{1/2} + \ln \frac{-\imagI w}{1+ \sqrt{1-w^2}}  
	=  (1-w^2)^{1/2} -  \ln \frac{\imagI ( 1+ \sqrt{1-w^2})}{w} \\
	& =  (1-w^2)^{1/2} -  \ln \lr{  w^{-1} - \imagI\sqrt{1-w^{-2} }   }    - \pi \imagI /2  
	\\
	& 
	=  (1-w^2)^{1/2}  +  \imagI  \arccos (1/w)   - \pi \imagI /2 \\
	&  = -\imagI  (w^2-1)^{1/2}  +  \imagI  \arccos (1/w)   - \pi \imagI /2.
\end{align*}
To verify the last equality, notice that $ \arg(w^2-1) \in [0; \pi ) $, thus  $ \arg (\imagI  (w^2-1)^{1/2})  \in [ \pi/2; \pi  )$ and $ \arg(\sqrt{1-w^2}) \in [  -\pi/2,  0 )  $, hence   $ \sqrt{1-w^2}  = -\imagI  \sqrt{w^2-1}  $.

Finally, to show  \eqref{eq:L01e02},  notice that  the function $ \frac{1}{12} + \frac{1}{6 \sqrt 5} + \lr{\frac{4}{27}}^{1/4} + \frac{c^2(c^2+2)}{\sqrt 8 (c^2-1)^{2.5}} $ is decreasing  in $ c $ for $ c>1 $ 	(seen by differentiating with respect to $ c $). Since  it takes value  $ 2.272365\ldots $ at $ c=\pi/2 $,   for $ \nu >1 $ and $ c =\Re(w)\geq \pi /2 $ we have
\[ 
\lrv{ 1+ \eta(\nu,-\imagI w)} \leq
1  + \exp \lr{ \frac{2 \mc V_{+\infty, - \imagI w}}{ \nu}  } \, \frac{2 \mc V_{+\infty, - \imagI w}}{ \nu} 
\leq 
1+2\cdot  2.273  \e^{2\cdot  2.273 }  <  430.
\]
\end{proof}
\begin{remark}\label{footnote:L8interp}
The  argument above about rewriting $ \xi(-\imagI w)   $, in effect, expresses $\xi(z)-z$ as $\imagI (g(\imagI z)-\pi/2)$ for $z : \arg z \in [-\pi/2,0)$, with $g $ defined as in \cref{thm:L06}. Thus, taking into account the symmetry of $\xi$ around the real axis, \cref{thm:L06} upper bounds  the real part of $z-\xi(z)$ for $z : \Re z>0,\, \lrv{\Im z} \geq 1$.
\end{remark}

\bibliographystyle{apsrev4-1}
\bibliography{biblio}

\end{document}